\documentclass[lettersize,journal]{IEEEtran}
\usepackage{amsmath,amsfonts,amssymb,amsthm}
\usepackage{algorithmic}
\usepackage{algorithm}
\usepackage{array}
\usepackage[caption=false,font=normalsize,labelfont=sf,textfont=sf]{subfig}
\usepackage{textcomp}
\usepackage{stfloats}
\usepackage{url}
\usepackage{verbatim}
\usepackage{graphicx}
\usepackage{cite}
\usepackage{booktabs}
\usepackage{multirow}
\usepackage{xcolor}
\usepackage{enumitem}
\usepackage[hidelinks]{hyperref}
\newtheorem{theorem}{Theorem}
\newtheorem{proposition}{Proposition}
\newtheorem{definition}{Definition}

\begin{document}

\title{Non-Degenerate Risk Certification for Automated Security Decisions: A Decision-Contract Theory with ATT\&CK-Aligned Triage as a Worked Instance}

\author{Zhenpeng~Li%
\IEEEcompsocitemizethanks{%
\IEEEcompsocthanksitem Zhenpeng Li is with Guangzhou Health Science College,
No.~248 Guangyuan Middle Road, Guangzhou, Guangdong 510405, China
(e-mail: 2025301001@gzws.edu.cn). Corresponding author: Zhenpeng Li.}}

\markboth{Journal of \LaTeX\ Class Files,~Vol.~14, No.~8, August~2021}%
{Shell \MakeLowercase{\textit{et al.}}: A Sample Article Using IEEEtran.cls for IEEE Journals}

\IEEEpubid{0000--0000/00\$00.00~\copyright~2021 IEEE}

\maketitle

\begin{abstract}
An unconditional risk bound on automated decisions can be satisfied without automating anything, since a selector that never acts drives the bound to zero.
We show this is structural: any risk certificate is defined over a decision contract, the inputs a system acts on plus the semantic relation under which an output counts correct, and weakening either hides base-classifier error.
We develop a decision-contract theory: an error-conservation law showing error is only reassigned among harmful automation, human deferral, and semantic masking; a label-free singleton capacity certifying structural incapacity, with a risk-feasible refinement separating recoverable threshold misalignment from risk-constrained incapacity; and a non-degenerate actionability certificate excluding all-abstain solutions by construction.
We instantiate this on ATT\&CK-aligned alert triage for LLM-based intrusion detection, the setting that exposed the vacuity failure.
Across 3 IDS datasets, 6 LLMs, and 4 error-rate thresholds, empirical false-attribution risk stays at or below target in 90.3\% of configurations, with 83.4\% mean correct automation.
The capacity diagnostic explains every low-utility configuration; its refinement separates genuine misalignment from risk-constrained incapacity, confirmed by an exhibited alternative threshold; a training-stability re-run finds no confirmed structural-incapacity instance; and real fine-grained attack-subtype labels confirm the coarsening-transfer identity under a genuine many-to-one map, with small but non-zero masking mass.
\end{abstract}


\section{Introduction}
\label{sec:intro}

Modern Security Operations Centers (SOCs) process thousands of intrusion detection system (IDS) alerts daily.
A critical step in incident response is attributing each alert to a specific attack technique within the MITRE ATT\&CK framework, which determines the appropriate containment and remediation actions.
Manual attribution is accurate but does not scale: the median SOC analyst handles 20--30 alerts per hour~\cite{alahmadi2022sok}, while enterprise IDS deployments generate orders of magnitude more.

Large language models (LLMs) have recently demonstrated strong performance on network traffic classification tasks~\cite{ferrag2024revolutionizing}, raising the prospect of automated ATT\&CK-aligned alert triage.
However, an LLM classifier can still return a confident but wrong technique label.
In the SOC context, a misattributed alert is worse than an unattributed one: it can trigger incorrect containment procedures, waste analyst time on the wrong investigation track, or mask a genuine threat behind a benign-looking attribution.

Existing approaches to LLM-based IDS offer no formal control over the misattribution rate.
Practitioners either accept all model outputs at face value (risking misattribution) or apply ad hoc confidence thresholds (which require manual tuning and lack guarantees).
Conformal risk control (CRC)~\cite{angelopoulos2022conformal} offers a principled alternative: wrap the classifier with a calibration-based abstention layer so that alerts whose nonconformity scores exceed a data-driven threshold are deferred, and the probability of a wrong automated attribution among the rest is controlled at a user-specified level~$\alpha$.

This risk statement, however, is weaker than it looks.
It is an unconditional bound: $R \leq \alpha$ holds trivially if the system abstains on everything.
We initially observed exactly this failure mode in one training run of an ML+CRC baseline, where a LightGBM classifier satisfied its FAR target by never automating a single alert; a subsequent stability check (Section~\ref{sec:exp-degenerate}) shows this specific instance does not reproduce across retrainings, but the underlying vacuity is a property of the risk bound itself, not of that one run, and is what motivates the theory developed next.
This is not a bug in one classifier's calibration; it is a structural consequence of what an unconditional risk bound can and cannot certify, one that applies to conformal risk control wherever it is used to gate automation, not only in this triage setting.
We show that risk can be driven down through exactly two mechanisms that have nothing to do with predictive competence---abstaining more (shrinking the set of inputs the system is responsible for) and reporting at a coarser semantic resolution (redefining some errors as correct)---and that these two mechanisms admit a common formal treatment.

Our contributions are:

\begin{enumerate}[leftmargin=*]
\item \textbf{A decision-contract theory of risk-certificate degeneracy.}
We formalize a decision contract as a selector-plus-correctness-relation pair, prove an error-conservation law under which a base classifier's error is only ever reassigned among harmful automation, human deferral, and semantic masking (Theorem~\ref{thm:error_conservation}), and derive an exact risk-transfer identity for semantic coarsening together with an impossibility result for its reverse direction (Theorems~\ref{thm:coarsening_transfer}--\ref{thm:reverse_impossibility}).

\item \textbf{Singleton capacity: a label-free deployability limit, and its risk-feasible refinement.}
For the singleton-set decision rule used throughout the paper, we derive an exact geometric characterization of automation (Proposition~\ref{prop:singleton_interval}) and a resulting capacity $\kappa(f)$ that upper-bounds the action rate attainable by \emph{any} global threshold (Theorem~\ref{thm:singleton_capacity}), with a finite-sample, unlabeled certificate (Proposition~\ref{prop:capacity_certificate}) and a formal counterexample ruling out score-variance-based diagnostics (Proposition~\ref{prop:variance_counterexample}).
$\kappa(f)$ is necessary but not sufficient to diagnose a shortfall as recoverable: we introduce the risk-feasible capacity $\kappa_\alpha(f)$ (Definition~\ref{def:risk_feasible_capacity}), which requires labels but correctly separates threshold misalignment from a risk-constrained limit that $\kappa(f)$ alone cannot see.

\item \textbf{A non-degenerate actionability certificate.}
We define $(\alpha,\rho)$-actionability, which jointly bounds harmful risk and floors the action rate, and show it converts an unconditional risk bound into both a correct-automation lower bound and a conditional-error upper bound (Theorem~\ref{thm:actionability_bridge})---a guarantee the risk bound alone cannot provide (Proposition~\ref{prop:risk_factorization}).
Unlike the risk side, the action-rate floor is not certified by exchangeability alone; we close this with a finite-sample Hoeffding lower bound on the deployed threshold's action rate (Proposition~\ref{prop:action_rate_certificate}).
This makes the certificate two-sided but not symmetric: the risk side keeps the standard marginal conformal guarantee of Theorem~\ref{thm:crc}, and the action-rate side receives a genuine $(1-\delta)$ finite-sample lower confidence bound for the already-realized $\hat\tau$---a different probabilistic statement, not one certified bound paired with an unqualified point estimate.

\item \textbf{Instantiation on ATT\&CK-aligned triage.}
We formalize ATT\&CK-aligned alert triage as a selective prediction problem and evaluate a CRC wrapper across 6 LLMs (7B--32B), 3 IDS benchmark datasets, and 4 error-rate thresholds.
At the representative deployment floor $\rho=0.5$ used throughout, the singleton-capacity diagnostic rules out structural incapacity for every LLM configuration (all 18, Table~\ref{tab:kappa}), and the risk-feasible refinement $\kappa_\alpha(f)$ (Definition~\ref{def:risk_feasible_capacity}) further separates the three lowest-utility cases into two that are risk-constrained incapacity and one genuine, recoverable threshold misalignment---a distinction $\kappa(f)$ alone gets wrong for two of the three; a 20-seed training-stability re-run of an ML+CRC baseline finds no confirmed instance of true structural incapacity, which we report as an open empirical question rather than a confirmed case (Section~\ref{sec:exp-degenerate}).
\end{enumerate}

\section{Related Work}
\label{sec:related}

\subsection{ATT\&CK Technique Mapping}

MITRE ATT\&CK provides a taxonomy of adversarial tactics and techniques observed in real-world intrusions.
Several systems map unstructured cyber threat intelligence (CTI) text to ATT\&CK techniques: rcATT~\cite{legoy2020automated} uses TF-IDF with multi-label classification, TRAM fine-tunes BERT-based models on CTI reports.
These systems operate on \emph{textual} input (threat reports, CVE descriptions) and solve a text classification problem.

Our work addresses a complementary task: mapping \emph{structured IDS alerts} (network flow features) to ATT\&CK techniques.
The input modality, inference pipeline, and label granularity differ fundamentally from CTI text mapping, and the two lines of work are not directly comparable.
What our approach adds is not a better mapping method, but a \emph{risk-controlled decision layer} that quantifies and bounds attribution uncertainty.

\subsection{Conformal Prediction for Safety-Critical Systems}

Conformal prediction~\cite{vovk2005algorithmic} constructs prediction sets with finite-sample coverage guarantees under exchangeability.
Angelopoulos et al.~\cite{angelopoulos2022conformal} generalized this to conformal risk control (CRC), enabling control of arbitrary monotone loss functions---including the false discovery rate and selective prediction risk---and subsequent work extends the same calibration principle to general risk-controlling algorithms via learn-then-test procedures~\cite{angelopoulos2022learn}.

Applications in safety-critical domains include medical diagnosis~\cite{lu2022fair}, autonomous driving~\cite{lindemann2023safe}, and natural language generation~\cite{quach2023conformal}.
In cybersecurity, conformal prediction has been applied to malware detection~\cite{jordaney2017transcend} and network anomaly detection~\cite{smith2015anomaly}, mostly for coverage-style guarantees.

\subsection{Selective Classification and the Object of Certification}
\label{sec:related_selective}

Selective classification~\cite{geifman2017selective,elyaniv2010foundations} equips a classifier with a reject option and characterizes the resulting risk--coverage trade-off; SelectiveNet~\cite{geifman2019selectivenet} integrates the reject decision into training to optimize this trade-off directly, and one-sided formulations~\cite{gangrade2021selective} give alternative selector constructions with their own coverage guarantees.
Decision-theoretic treatments of selective labeling~\cite{wei2021selective} independently arrive at a threshold policy as the optimal selector under an explicit cost model---a policy-optimization question our decision-contract framing does not address, since we take the selector as given by the conformal calibration procedure and ask instead what a certificate over that selector actually establishes.
This literature already establishes that a risk bound alone is insufficient without a coverage (action-rate) guarantee alongside it---the same tension that motivates Proposition~\ref{prop:risk_factorization} and the actionability certificate of Section~\ref{sec:theory_contract}.
Recent work on LLM outputs specifically shows that simple confidence statistics are unreliable selection signals~\cite{phillips2026entropy}, echoing Proposition~\ref{prop:variance_counterexample}'s formal argument that true-label score statistics do not determine singleton capacity, and that risk control can fail to certify what practitioners actually need from structured LLM generations~\cite{kotte2026crc}, an impossibility-flavored result adjacent to but distinct from Theorem~\ref{thm:reverse_impossibility}: theirs concerns what CRC can certify about structured outputs in general, ours concerns what a coarsening of a fixed selector's outputs can certify about the selector's finer-grained behavior.
Selective risk control for LLM action policies is an active target for conformal methods more broadly~\cite{khosravi2026csa}.

Our contribution is not that risk and coverage both matter, which selective classification already shows; it is a generalization of \emph{what is being certified}.
The selective-classification line of work cited above studies the selector: when to abstain, at a fixed correctness relation.
We generalize the certified object to a decision contract $(g,\mathcal{E})$, in which both the selector \emph{and} the correctness relation can vary, and show that this second axis---semantic coarsening---hides base-classifier error by exactly the same accounting as abstention (Theorem~\ref{thm:error_conservation}).
This is why a coverage guarantee at one label resolution does not, by itself, answer whether a report at a coarser resolution is trustworthy: Theorem~\ref{thm:coarsening_transfer} and its reverse-impossibility counterpart (Theorem~\ref{thm:reverse_impossibility}) give an exact accounting and a limit for that question specifically, complementary to work on varying prediction specificity in a label hierarchy rather than a substitute for it.
Class-conditional conformal prediction~\cite{ding2023class} addresses a related but distinct concern---per-class coverage can be masked by a favorable marginal average across many classes---and our within-fiber confusion mass $M_\phi(g)$ (Theorem~\ref{thm:coarsening_transfer}) is the analogous accounting object for masking induced by coarsening the label space itself rather than by aggregating coverage across existing classes.
Conformal risk control~\cite{angelopoulos2022conformal} gives our risk side its finite-sample guarantee; selective classification's risk--coverage lens gives our action-rate side its natural counterpart (Proposition~\ref{prop:action_rate_certificate}); the decision-contract formalism is what places both, plus semantic coarsening, under a single certified object.

\subsection{LLMs for Network Intrusion Detection}

Recent work has explored LLMs as network intrusion detectors, serializing flow features into natural-language prompts for classification~\cite{ferrag2024revolutionizing}.
These approaches achieve competitive accuracy with traditional ML models (XGBoost, LightGBM) on standard benchmarks, but inherit the calibration challenges of LLMs: softmax probabilities are often poorly calibrated, and confident errors are common.

Our framework is model-agnostic: it wraps any LLM-based (or ML-based) classifier and provides post-hoc FAR control without modifying the base model.
This composability is a practical advantage, since the abstention layer can be deployed on top of existing LLM-IDS pipelines.

\section{Problem Formulation}
\label{sec:formulation}

\subsection{Task Definition}

Let $\mathcal{X}$ denote the space of network flow feature vectors and $\mathcal{T} = \{t_1, \ldots, t_K\}$ denote a set of ATT\&CK techniques.
An IDS alert $x \in \mathcal{X}$ has a benchmark-derived, ATT\&CK-aligned operational technique label $y(x) \in \mathcal{T}$, obtained by mapping the dataset's own attack-category annotation through $\phi$ (Section~\ref{sec:mapping}); this is not an independently sourced ATT\&CK ground-truth annotation, and we use $y(x)$ throughout in this operational sense.
A classifier $f: \mathcal{X} \to \Delta^{K}$ produces a probability vector over techniques, where $f_k(x)$ estimates $P(y = t_k \mid x)$.

We consider the \emph{selective prediction} setting: the system either outputs a single technique attribution $\hat{t}(x) \in \mathcal{T}$, or abstains ($\hat{t}(x) = \bot$), deferring the alert to a human analyst.

\subsection{ATT\&CK Label Mapping}
\label{sec:mapping}

Standard IDS benchmark datasets use attack-category labels (e.g., DoS, Probe, Exploitation) rather than ATT\&CK identifiers.
We define a deterministic bijective mapping $\phi: \mathcal{C}_{\text{attack}} \to \mathcal{T}$ from attack categories to ATT\&CK techniques, shown in Table~\ref{tab:mapping}.
Since $\phi$ is a bijection, the classification problem over $\mathcal{C}_{\text{attack}}$ and the attribution problem over $\mathcal{T}$ are equivalent: $P(y = t_k \mid x) = P(c = \phi^{-1}(t_k) \mid x)$.

\begin{table*}[t]
\centering
\caption{Deterministic mapping from IDS attack categories to ATT\&CK techniques. Each attack category maps to exactly one technique (bijection), preserving the classification problem structure.}
\label{tab:mapping}
\begin{tabular}{llll}
\toprule
IDS Label & ATT\&CK Tactic & Technique ID & Mapping Basis \\
\midrule
DoS & Impact & T1498 & Volumetric traffic patterns \\
CredentialAccess & Credential Access & T1110 & Authentication failure rate \\
Exploitation & Initial Access & T1190 & Exploit payload signatures \\
Probe & Discovery & T1046 & Port scanning behavior \\
\bottomrule
\end{tabular}
\end{table*}

This mapping is intentionally coarse-grained (4 techniques from 4 tactics), reflecting the label granularity available in existing IDS benchmarks.
We discuss the implications of this granularity in Section~\ref{sec:limitations}.

\subsection{Nonconformity Score Design}

We define the nonconformity score for sample $x$ and candidate technique $t_k$ as:
\begin{equation}
\label{eq:nc_score}
s(x, t_k) = 1 - \frac{f_k(x)}{\sum_{j: t_j \neq \text{Normal}} f_j(x)},
\end{equation}
where the denominator normalizes over attack classes only (excluding the Normal/benign class).
A low $s(x, t_k)$ indicates high confidence that $x$ corresponds to technique $t_k$.

The normalization in Eq.~\eqref{eq:nc_score} serves two purposes: (i) it removes the influence of the Normal-class probability, which is irrelevant for technique attribution among known-attack alerts; and (ii) it improves cross-model comparability, since raw softmax magnitudes vary across LLM architectures.

\subsection{Risk Guarantee}

\begin{definition}[False Automated Attribution Risk]
For a selective classifier that outputs $\hat{t}(x) \in \mathcal{T} \cup \{\bot\}$, we define:
\begin{equation}
\text{FAR} = P\bigl(\hat{t}(x) \neq \bot,\; \hat{t}(x) \neq y(x)\bigr).
\end{equation}
FAR is therefore an unconditional risk over all attack alerts: it counts only alerts that are automatically attributed to the wrong technique.
For diagnostics, we also report the conditional automated-attribution error, $P(\hat{t}(x) \neq y(x) \mid \hat{t}(x) \neq \bot)$, but this conditional quantity is not the formal target controlled by the conformal threshold.
\end{definition}

\begin{theorem}[Conformal bound for wrong automated attribution]
\label{thm:crc}
Let $\mathcal{D}_{\text{cal}} = \{(x_i, y_i)\}_{i=1}^{n}$ be an exchangeable calibration set.
Define the threshold:
\begin{equation}
\hat{q} = \inf\left\{q : \frac{1}{n+1}\sum_{i=1}^{n} \mathbf{1}[s(x_i, y_i) \leq q] \geq 1 - \alpha\right\},
\end{equation}
and the conformal prediction set:
\begin{equation}
\Gamma_{\hat{q}}(x) = \{t_k \in \mathcal{T}: s(x,t_k) \leq \hat{q}\}.
\end{equation}
The selective decision rule is:
\begin{equation}
\hat{t}(x) = \begin{cases}
\text{the unique element of } \Gamma_{\hat{q}}(x) & \text{if } |\Gamma_{\hat{q}}(x)| = 1, \\
\bot & \text{otherwise}.
\end{cases}
\end{equation}
Then the probability of a wrong automated attribution is bounded by the conformal miscoverage risk:
\begin{equation}
P\bigl(\hat{t}(x) \neq \bot,\; \hat{t}(x) \neq y(x)\bigr)
\leq
P\bigl(y(x) \notin \Gamma_{\hat{q}}(x)\bigr)
\leq \alpha,
\end{equation}
up to the standard finite-sample quantile convention.
\end{theorem}

The guarantee is distribution-free in the usual conformal sense: it requires no parametric assumptions on $P(X,Y)$, only exchangeability between calibration and test samples.
The bound applies to the unconditional probability of a wrong automated attribution.
It does not imply that every realized split has empirical FAR below $\alpha$, nor does it directly bound the conditional error rate among non-abstained predictions.
We therefore treat ``empirical FAR $\leq \alpha$ across all random seeds'' as a stringent operational check, and report conditional automated-attribution error only as a diagnostic quantity.

\subsection{ATT\&CK Hierarchy Guarantee}

ATT\&CK organizes techniques under tactics (e.g., T1498 $\subset$ Impact).
Since our mapping $\phi$ assigns each technique to a unique tactic, correct technique attribution implies correct tactic attribution.
Formally, if $\psi: \mathcal{T} \to \mathcal{A}$ maps techniques to tactics (a surjection in general, a bijection in our 4-class setting), then:
\begin{equation}
\hat{t}(x) = y(x) \implies \psi(\hat{t}(x)) = \psi(y(x)).
\end{equation}
Therefore, the FAR bound at the technique level automatically propagates to the tactic level.

Theorem~\ref{thm:crc} and this hierarchy propagation are standard consequences of split conformal prediction, restated here for the ATT\&CK attribution setting.
They establish that FAR can be controlled; they do not establish that a controlled FAR is operationally meaningful.
Section~\ref{sec:theory_contract} shows that an unconditional bound of this form can be satisfied without any automation at all, and develops the additional conditions under which a low FAR corresponds to genuine, non-degenerate triage.

\section{Decision-Contract Theory for Non-Degenerate Risk Certification}
\label{sec:theory_contract}

The central premise of this section is that a risk bound such as Theorem~\ref{thm:crc} is not, by itself, a property of a good predictor.
It is a property of a \emph{decision contract}: the set of alerts on which the system agrees to act, and the semantic relation under which an automated output is counted as correct.
A system can reduce its measured FAR without improving its underlying predictive competence, either by abstaining more often or by reporting at a coarser semantic resolution.
We formalize these two mechanisms, show that they are the only two mechanisms by which the base classifier's fine-grained error can be hidden from the risk bound, and use the result to define a deployment certificate that cannot be satisfied by either mechanism alone.

\subsection{Decision contracts}

Let $\mathcal{Y} = \mathcal{T}$ be the finest available label space (Section~\ref{sec:formulation}) and let $h: \mathcal{X} \to \mathcal{Y}$ denote the top-one prediction of the base classifier $f$.
A selector $g: \mathcal{X} \to \{0,1\}$ determines whether the system acts: $g(x)=1$ recovers $\hat{t}(x) \neq \bot$, and $g(x)=0$ recovers $\hat{t}(x)=\bot$.
Correctness is represented by a reflexive acceptance relation $\mathcal{E} \subseteq \mathcal{Y}\times\mathcal{Y}$.
The finest relation, matching Theorem~\ref{thm:crc}, is $\mathcal{E}_{\mathrm{fine}} = \{(y,y): y\in\mathcal{Y}\}$.
A coarser contract is induced by the tactic map $\psi$ of Section~\ref{sec:formulation}: $\mathcal{E}_{\psi} = \{(y,y'): \psi(y)=\psi(y')\}$.
We use a relation rather than a function so that the same object accommodates ATT\&CK structures in which a technique belongs to more than one tactic.

\begin{definition}[Decision contract]
\label{def:contract}
A decision contract is a pair $C=(g,\mathcal{E})$.  Its harmful automated risk, action rate, accepted utility, and conditional error are
\begin{align}
R(C) &= \mathbb{P}\bigl(g(X)=1,\ (Y,h(X))\notin\mathcal{E}\bigr), \\
A(C) &= \mathbb{P}\bigl(g(X)=1\bigr), \\
U(C) &= \mathbb{P}\bigl(g(X)=1,\ (Y,h(X))\in\mathcal{E}\bigr) = A(C)-R(C), \\
Q(C) &= \mathbb{P}\bigl((Y,h(X))\notin\mathcal{E} \mid g(X)=1\bigr),
\end{align}
where $Q(C)$ is defined when $A(C)>0$.  Under $\mathcal{E}=\mathcal{E}_{\mathrm{fine}}$, $R(C)$ is exactly the FAR of Theorem~\ref{thm:crc} and $A(C)=1-\text{AR}$.
\end{definition}

\begin{proposition}[Risk factorization and vacuity]
\label{prop:risk_factorization}
For every decision contract with $A(C)>0$, $R(C)=A(C)\,Q(C)$.
Consequently, a bound $R(C)\leq\alpha$ alone implies neither a positive action rate nor a small conditional error: $g\equiv 0$ gives $R(C)=0$ trivially, and a contract with $A(C)\leq\alpha$ may have $Q(C)=1$ while still satisfying $R(C)\leq\alpha$.
\end{proposition}
\begin{proof}
The identity is the product rule for probability; the two constructions instantiate it at the boundary $A=0$ and at $Q=1$.
\end{proof}

Proposition~\ref{prop:risk_factorization} is not a hypothetical concern.  Section~\ref{sec:exp-degenerate} exhibits a real ML+CRC configuration in which $A(C)=0$ and the FAR bound is satisfied by construction.

\subsection{Contract weakening and an error-conservation law}

\begin{definition}[Contract order]
$C_1=(g_1,\mathcal{E}_1) \succeq C_2=(g_2,\mathcal{E}_2)$ if $g_1(x)\geq g_2(x)$ almost everywhere and $\mathcal{E}_1\subseteq\mathcal{E}_2$: the stronger contract acts on at least as many alerts and accepts no more semantic errors as correct.
\end{definition}

\begin{theorem}[Risk monotonicity under contract weakening]
\label{thm:contract_monotonicity}
If $C_1\succeq C_2$, then $R(C_2)\leq R(C_1)$, and the reduction decomposes exactly as
\begin{align}
R(C_1)-R(C_2) ={}& \mathbb{P}\bigl(g_1{=}1,g_2{=}0,(Y,h)\notin\mathcal{E}_1\bigr) \notag\\
&+ \mathbb{P}\bigl(g_2{=}1,(Y,h)\notin\mathcal{E}_1,(Y,h)\in\mathcal{E}_2\bigr),
\label{eq:weakening_decomposition}
\end{align}
where the first term is risk removed by additional deferral and the second is risk removed by semantic relaxation.
\end{theorem}
\begin{proof}
Since $g_2\leq g_1$ and $\mathcal{E}_1\subseteq\mathcal{E}_2$, the pointwise difference of risk indicators equals $(g_1-g_2)\mathbf{1}\{(Y,h)\notin\mathcal{E}_1\} + g_2\mathbf{1}\{(Y,h)\notin\mathcal{E}_1,(Y,h)\in\mathcal{E}_2\}$, which is non-negative; take expectations.
\end{proof}

A direct consequence is an accounting identity for the base classifier's fine-grained error.  Define $B(h)=\mathbb{P}(h(X)\neq Y)$, deferred error $D(C)=\mathbb{P}(g{=}0,h\neq Y)$, semantic masking $M(C)=\mathbb{P}(g{=}1,h\neq Y,(Y,h)\in\mathcal{E})$, and genuine correct automation $G(C)=\mathbb{P}(g{=}1,h{=}Y)$.

\begin{theorem}[Error conservation under a decision contract]
\label{thm:error_conservation}
For every reflexive $\mathcal{E}\supseteq\mathcal{E}_{\mathrm{fine}}$,
\begin{equation}
B(h) = R(C) + D(C) + M(C), \qquad U(C) = G(C) + M(C).
\label{eq:error_conservation}
\end{equation}
\end{theorem}
\begin{proof}
On $\{h(X)\neq Y\}$, exactly one of three disjoint events occurs: the system defers ($D$); the system acts and $\mathcal{E}$ still rejects the error ($R$); the system acts and $\mathcal{E}$ accepts the error ($M$).  Among acted-upon cases accepted by $\mathcal{E}$, either $h=Y$ or the case is a masked error, giving the second identity.
\end{proof}

Eq.~\eqref{eq:error_conservation} is the paper's central accounting result: a fixed base classifier's error $B(h)$ does not disappear under a risk certificate, it is only reassigned among harmful automation, human deferral, and semantic masking.  A reported reduction in risk, or a reported gain in utility from coarsening (e.g., technique- to tactic-level reporting), is therefore only informative once $M(C)$ is reported alongside it.

\subsection{Risk transfer under semantic coarsening}
\label{sec:coarsening_theory}

Let $\phi:\mathcal{Y}\to\mathcal{Z}$ be a deterministic, possibly many-to-one map, and fix the same selector $g$ at both resolutions.

\begin{theorem}[Exact fine-to-coarse risk transfer]
\label{thm:coarsening_transfer}
$R_{\phi}(g) = R_{\mathrm{fine}}(g) - M_{\phi}(g)$, where $M_{\phi}(g)=\sum_{z}\sum_{y\neq y' \in \phi^{-1}(z)} \mathbb{P}(g{=}1,Y{=}y,h(X){=}y')$ is the within-fiber confusion mass.  Hence $R_{\phi}(g)\leq R_{\mathrm{fine}}(g)$, with equality iff $M_\phi(g)=0$.
\end{theorem}
\begin{proof}
A fine-label error stops being an error after coarsening exactly when the true and predicted labels share a fiber of $\phi$; subtract this disjoint mass from $R_{\mathrm{fine}}$.
\end{proof}

\begin{theorem}[Impossibility of reverse risk transfer]
\label{thm:reverse_impossibility}
If some fiber of $\phi$ contains two distinct labels $y\neq y'$, then $R_\phi(g)=0$ does not imply any non-trivial bound on $R_{\mathrm{fine}}(g)$: taking $g\equiv1$, $Y=y$ a.s., $h(X)=y'$ a.s.\ gives $R_\phi(g)=0$ and $R_{\mathrm{fine}}(g)=1$.
\end{theorem}

In our 4-technique, 4-tactic mapping (Table~\ref{tab:mapping}), every fiber of $\psi$ is a singleton, so $M_\psi(g)\equiv 0$ identically: the technique-to-tactic identity of Section~\ref{sec:formulation}, and the empirical equality reported in Section~\ref{sec:rq4}, are the zero-masking boundary case of Theorem~\ref{thm:coarsening_transfer}, not an independent finding.
Theorem~\ref{thm:reverse_impossibility} is stated for the general many-to-one case relevant to the full ATT\&CK matrix, where multiple techniques share a tactic; evaluating $M_\phi(g)>0$ empirically requires technique-level labels finer than the 4-category benchmarks used here (Section~\ref{sec:limitations}), and we leave a direct empirical demonstration of non-trivial masking to future work on finer-grained taxonomies.

\subsection{Exact geometry of singleton automation}
\label{sec:singleton_geometry}

Write $p_k(x) = f_k(x)/\sum_{j\neq\mathrm{Normal}} f_j(x)$ so that $s(x,t_k)=1-p_k(x)$ (Eq.~\eqref{eq:nc_score}) and $\tau=1-\hat q$.  The prediction set is $\Gamma_\tau(x)=\{k: p_k(x)\geq\tau\}$.  Let $p_{(1)}(x)\geq p_{(2)}(x)$ be the largest and second-largest attack-normalized probabilities.

\begin{proposition}[Singleton interval characterization]
\label{prop:singleton_interval}
$|\Gamma_\tau(x)|=1 \iff p_{(2)}(x)<\tau\leq p_{(1)}(x)$.  Hence every $x$ defines a singleton interval $I(x)=(p_{(2)}(x),p_{(1)}(x)]$, and a global threshold automates $x$ iff $\tau\in I(x)$.
\end{proposition}

Let $F_j^-(t)=\mathbb{P}(p_{(j)}(X)<t)$ for $j\in\{1,2\}$.

\begin{theorem}[Singleton capacity]
\label{thm:singleton_capacity}
The global-threshold action rate is $A(\tau)=F_2^-(\tau)-F_1^-(\tau)$.  The \emph{singleton capacity}
\begin{equation}
\kappa(f) = \sup_{\tau\in[0,1]} \bigl(F_2^-(\tau)-F_1^-(\tau)\bigr)
\label{eq:singleton_capacity}
\end{equation}
is the maximum action rate attainable by any global singleton threshold.  If $\kappa(f)<\rho$, no calibration method using a single global threshold can satisfy an action-rate requirement $A\geq\rho$, regardless of $\alpha$ or calibration-set size.
\end{theorem}
\begin{proof}
$\mathbb{P}(p_{(2)}<\tau\leq p_{(1)}) = \mathbb{P}(p_{(2)}<\tau)-\mathbb{P}(p_{(1)}<\tau)$ since $\{p_{(1)}<\tau\}\subseteq\{p_{(2)}<\tau\}$; take the supremum over $\tau$.
\end{proof}

$\kappa(f)$ is unconstrained by risk: it is the largest action rate reachable by \emph{some} threshold, without regard to how many of the resulting singleton predictions are wrong.  It therefore yields a one-directional diagnosis.  If $\kappa(f)<\rho$, no threshold---calibrated or otherwise---can reach the action-rate requirement $A\geq\rho$ at \emph{any} risk level, and this is certified from the classifier's probability output $f$ alone, without labels: \emph{structural incapacity}.  The converse does not hold: $\kappa(f)\geq\rho$ rules out this particular geometric obstruction, but it does not imply that a threshold exists which is simultaneously risk-feasible ($R(\tau)\leq\alpha$) and reaches $\rho$, because the action-maximizing $\tau$ in Eq.~\eqref{eq:singleton_capacity} need not be the same $\tau$ that keeps risk below $\alpha$.

\begin{definition}[Risk-feasible capacity]
\label{def:risk_feasible_capacity}
$\kappa_\alpha(f) = \sup\{A(\tau) : \tau\in[0,1],\ R(\tau)\leq\alpha\}$, the largest action rate attainable among thresholds that also satisfy the risk target.
\end{definition}

By definition $\kappa_\alpha(f)\leq\kappa(f)$, and $\kappa(f)<\rho \implies \kappa_\alpha(f)<\rho$, so structural incapacity certified via $\kappa(f)$ remains valid without reference to $\kappa_\alpha$.  Unlike $\kappa(f)$, $\kappa_\alpha(f)$ requires labeled data to evaluate $R(\tau)$ and is not a pre-deployment, label-free quantity.  We use the two together to distinguish three regimes rather than two: \emph{structural incapacity} ($\kappa(f)<\rho$, no threshold at any risk level reaches $\rho$); \emph{risk-constrained incapacity} ($\kappa(f)\geq\rho$ but $\kappa_\alpha(f)<\rho$, the geometry admits a high-action threshold only at risk levels the target $\alpha$ forbids); and \emph{threshold misalignment} ($\kappa_\alpha(f)\geq\rho$ but the calibrated $\hat\tau$ has $A(\hat\tau)<\rho$, a genuinely recoverable shortfall).  The \emph{actionability gap} $\mathrm{AG}(\hat\tau)=\kappa(f)-A(\hat\tau)$ remains a useful label-free upper bound on how much automation could conceivably be at stake, but Section~\ref{sec:failure} shows it routinely overstates what is risk-feasible; $\kappa_\alpha(f)-A(\hat\tau)$ is the quantity that actually determines whether recalibration can close the gap.

\begin{proposition}[Finite-sample capacity certificate]
\label{prop:capacity_certificate}
For an unlabeled readiness sample of size $m$, the empirical gap estimator $\widehat\kappa=\sup_\tau(\widehat F_2^- - \widehat F_1^-)$ satisfies, by the DKW inequality and a union bound over the two CDFs, $|\widehat\kappa-\kappa(f)|\leq 2\sqrt{\log(4/\delta)/(2m)}$ with probability at least $1-\delta$.  If $\widehat\kappa$ plus this margin falls below $\rho$, structural incapacity is certified without any test labels.  This certificate applies to $\kappa(f)$ only; $\kappa_\alpha(f)$ has no label-free analogue, since it is defined through the risk constraint $R(\tau)\leq\alpha$.
\end{proposition}

\begin{proposition}[True-label score statistics do not determine capacity]
\label{prop:variance_counterexample}
Neither the mean nor the variance of the true-label nonconformity score, alone or combined with top-one accuracy, determines $\kappa(f)$.
\end{proposition}
\begin{proof}
Let $K=3$, $Y=1$ a.s.  Predictor A outputs $(0.3,0.7,0)$ on every input; predictor B outputs $(0.3,0.35,0.35)$.  Both have true-label score $1-0.3=0.7$ (equal mean, zero variance) and zero top-one accuracy.  For A, $\tau\in(0.3,0.7]$ singletons every input, so $\kappa=1$.  For B, the top two probabilities are tied for every input, so no $\tau$ yields a singleton and $\kappa=0$.
\end{proof}

Proposition~\ref{prop:variance_counterexample} is a formal argument against the heuristic used in an earlier draft of this analysis (nonconformity-score variance combined with base accuracy) to diagnose deployability; Section~\ref{sec:failure} replaces that heuristic with $\kappa(f)$ as a label-free necessary condition and $\kappa_\alpha(f)$ (Definition~\ref{def:risk_feasible_capacity}) as the labeled quantity that determines whether a shortfall is actually recoverable by recalibration.

\subsection{Operational risk and surrogate slack}

Let $e(X,Y)=\mathbf{1}\{h(X)\neq Y\}$, $g_\tau(X)=\mathbf{1}\{p_{(2)}(X)<\tau\leq p_{(1)}(X)\}$, and $C(\tau)=\mathbb{P}(Y\notin\Gamma_\tau(X))=\mathbb{P}(p_Y(X)<\tau)$ the miscoverage event controlled directly by Theorem~\ref{thm:crc}.

\begin{theorem}[Surrogate-slack decomposition]
\label{thm:surrogate_slack}
For the singleton decision rule, $C(\tau) = R(\tau) + S(\tau)$, where $S(\tau)=\mathbb{P}(g_\tau{=}0,\,Y\notin\Gamma_\tau(X))$ is the deferred-miscoverage mass.
\end{theorem}
\begin{proof}
Partition the miscoverage event by $g_\tau$.  When $g_\tau=1$ the set is a singleton, so the true label is excluded iff the unique output is wrong, giving $R(\tau)$; the remainder is $S(\tau)$.
\end{proof}

Conformal miscoverage control is therefore a valid but generally conservative surrogate for wrong automated attribution: the conservatism is exactly $S(\tau)$, the risk budget spent on alerts that are already deferred and hence operationally harmless.  Section~\ref{sec:exp-slack} measures $S(\tau)$ directly on the calibrated B2/B3 configurations.

\subsection{Non-degenerate actionability certificates}

\begin{definition}[$(\alpha,\rho)$-actionable contract]
\label{def:actionable}
$C=(g,\mathcal{E}_{\mathrm{req}})$ is $(\alpha,\rho)$-actionable if $R(C)\leq\alpha$ and $A(C)\geq\rho$ for a required action-rate floor $\rho>0$.
\end{definition}

\begin{theorem}[Bridge from unconditional risk to useful automation]
\label{thm:actionability_bridge}
If $C$ is $(\alpha,\rho)$-actionable, then $U(C)\geq\rho-\alpha$ and $Q(C)\leq\alpha/\rho$.
\end{theorem}
\begin{proof}
$U=A-R\geq\rho-\alpha$ directly; $Q=R/A\leq\alpha/\rho$ by Proposition~\ref{prop:risk_factorization}.
\end{proof}

Theorem~\ref{thm:actionability_bridge} excludes the all-abstain solution by construction: e.g., $\alpha=0.05,\rho=0.80$ jointly certify at least 75\% correct automation and at most 6.25\% conditional error among automated outputs — a guarantee that an unconditional FAR bound alone cannot provide (Proposition~\ref{prop:risk_factorization}).

Definition~\ref{def:actionable} and Theorem~\ref{thm:actionability_bridge} are population-level statements: $R(C)\leq\alpha$ is certified at deployment by the conformal calibration of Theorem~\ref{thm:crc}, but $A(C)\geq\rho$ as written requires knowing the true population action rate, which a practitioner does not have.  Checking $A(\hat\tau)\geq\rho$ against a single realized test set, as an empirical measurement, is not itself a certificate: it carries no stated confidence level and can be violated by an unlucky sample.  We close this gap with a finite-sample bound on the deployed threshold's action rate, giving a certificate whose action-rate side is exchangeability-based in the same sense as its risk side.

\begin{proposition}[Finite-sample action-rate certificate]
\label{prop:action_rate_certificate}
Let $\hat\tau$ be fixed independently of a certification sample $\{X_i\}_{i=1}^{m}$ drawn from the deployment distribution, and let $\widehat A(\hat\tau)=\frac{1}{m}\sum_i \mathbf{1}\{g_{\hat\tau}(X_i)=1\}$.  For $\varepsilon_m(\delta)=\sqrt{\log(1/\delta)/(2m)}$, the one-sided Hoeffding bound gives, with probability at least $1-\delta$,
\begin{equation}
A(\hat\tau) \;\geq\; \underline{A}_\delta \;:=\; \max\{0,\ \widehat A(\hat\tau)-\varepsilon_m(\delta)\}.
\label{eq:action_rate_lcb}
\end{equation}
If, in addition, $R(C_{\hat\tau})\leq\alpha$ (Theorem~\ref{thm:crc}'s standard marginal conformal guarantee) and $\underline{A}_\delta\geq\rho$, then substituting $\underline A_\delta$ for $\rho$ in Theorem~\ref{thm:actionability_bridge} gives $U(C_{\hat\tau})\geq\underline A_\delta-\alpha$ and $Q(C_{\hat\tau})\leq\alpha/\underline A_\delta$ on the probability-$(1-\delta)$ event $\{A(\hat\tau)\geq\underline A_\delta\}$.
\end{proposition}
\begin{proof}
$\mathbf{1}\{g_{\hat\tau}(X_i)=1\}$ are i.i.d.\ Bernoulli$(A(\hat\tau))$ since $\hat\tau$ is fixed before the certification draw; the one-sided Hoeffding inequality gives $\mathbb{P}(A(\hat\tau)<\widehat A(\hat\tau)-\varepsilon_m(\delta))\leq\delta$, which rearranges to Eq.~\eqref{eq:action_rate_lcb}.  Substituting $\underline A_\delta$ for $\rho$ in Theorem~\ref{thm:actionability_bridge}'s proof, which uses only $A(C)\geq\rho$, gives the stated bounds on the same probability-$(1-\delta)$ event.
\end{proof}

The two sides of the certificate carry genuinely different guarantees, and we do not combine them into a single confidence level.  $R(C_{\hat\tau})\leq\alpha$ is the standard split-conformal marginal statement of Theorem~\ref{thm:crc}: a probability over the joint draw of calibration set and test point, not a $(1-\delta)$-probability statement about the population risk of this specific realized $\hat\tau$.  $A(\hat\tau)\geq\underline A_\delta$ is a genuine finite-sample, $(1-\delta)$-confidence lower bound for the fixed, already-realized $\hat\tau$, since $\hat\tau$ is held constant while $\{X_i\}$ is redrawn.  Reporting both bounds together gives a practitioner two guarantees of different character, not one guarantee at a single stated level; a training-conditional or PAC-style high-probability version of the risk side, combined with $\delta_R+\delta_A$ union-bounding, would be needed to state a single joint confidence level, which we leave for future work.

Proposition~\ref{prop:action_rate_certificate} converts $(\alpha,\rho)$-actionability from a definition that a realized action rate happens to satisfy into a certificate a practitioner can compute before trusting a deployment: fix $\hat\tau$ on a calibration split, measure $\widehat A(\hat\tau)$ on a disjoint certification split, and report $\underline A_\delta$ rather than the raw empirical rate.  We compute this bound for the paper's 18 LLM configurations in Section~\ref{sec:exp-certificate}.

\subsection{Testable predictions}
\label{sec:theory_predictions}

The framework yields pre-specified predictions checked against existing D2 data with no new inference runs, except where noted:
\begin{enumerate}[leftmargin=*,itemsep=2pt]
\item[\textbf{P1}] The observed automation rate equals $\widehat F_2^-(\tau)-\widehat F_1^-(\tau)$ up to numerical ties (Theorem~\ref{thm:singleton_capacity}).
\item[\textbf{P2}] $\mathrm{Utility}(\alpha) \leq \kappa(f)$ for every configuration and $\alpha$.
\item[\textbf{P3}] Low-utility configurations are not structural incapacity ($\kappa(f)$ near zero); Definition~\ref{def:risk_feasible_capacity}'s risk-feasible capacity $\kappa_\alpha(f)$, not the unconstrained $\kappa(f)$, determines whether the remaining shortfall is genuine threshold misalignment ($\kappa_\alpha(f)\geq\rho$, recoverable by recalibration) or risk-constrained incapacity ($\kappa_\alpha(f)<\rho$, not recoverable at the given $\alpha$); a genuine structural-incapacity instance, if one occurs in this benchmark suite, requires a dedicated per-sample re-run of the ML pipeline and a training-stability check to confirm it is not an artifact (Section~\ref{sec:exp-degenerate}).
\item[\textbf{P4}] $C(\tau)-R(\tau)=S(\tau)\geq 0$ across all 18 LLM configurations, not only the single illustrative configuration of Section~\ref{sec:rq3}.
\item[\textbf{P5}] Under a genuinely many-to-one $\phi$ (real fine-grained attack subtypes collapsed onto the 4-category ATT\&CK mapping), $R_\phi(g) = R_{\mathrm{fine}}(g) - M_\phi(g)$ holds exactly, with $M_\phi(g) > 0$ (Section~\ref{sec:coarsening_theory}).
\item[\textbf{P6}] Configurations whose certified action-rate lower bound $\underline A_\delta$ (Proposition~\ref{prop:action_rate_certificate}) meets a floor $\rho$, jointly with $R\leq\alpha$, satisfy $U\geq\underline A_\delta-\alpha$ and $Q\leq\alpha/\underline A_\delta$; configurations with low FAR but near-zero action rate fail the certificate by construction, and the certificate is now stated at a declared confidence level rather than against a single realized test set.
\item[\textbf{P7}] $B(h) = R(C) + D(C)$ (the $\mathcal{E}=\mathcal{E}_{\mathrm{fine}}$ special case of Theorem~\ref{thm:error_conservation}, where $M(C)\equiv 0$) holds exactly for every configuration and $\alpha$.
\end{enumerate}

\section{Instantiating the Decision Contract on ATT\&CK-Aligned Triage}
\label{sec:method}

We now instantiate the framework of Section~\ref{sec:theory_contract} on the ATT\&CK attribution task of Section~\ref{sec:formulation}.
The deployed decision contract is $C_{\hat\tau}=(g_{\hat\tau},\mathcal{E}_{\mathrm{fine}})$: the selector $g_{\hat\tau}$ is the singleton rule of Proposition~\ref{prop:singleton_interval} at threshold $\hat\tau=1-\hat q$, and correctness is judged at the finest ATT\&CK technique resolution (Section~\ref{sec:rq4} instantiates the coarser tactic-level contract $\mathcal{E}_\psi$ as a comparison point).
Calibrating $\hat q$ (Algorithm~\ref{alg:crc}) certifies $R(C_{\hat\tau})\leq\alpha$ by Theorem~\ref{thm:crc}; Section~\ref{sec:results} additionally reports $A(C_{\hat\tau})$, $\kappa(f)$, the actionability gap, and the certificate of Proposition~\ref{prop:action_rate_certificate} for every configuration, rather than treating $R(C_{\hat\tau})\leq\alpha$ as sufficient on its own.

\subsection{CRC Calibration Procedure}

Figure~\ref{fig:architecture} illustrates the calibration and inference pipeline that produces $g_{\hat\tau}$, and Algorithm~\ref{alg:crc} gives the procedure in full.
The procedure requires only a held-out calibration set of labeled attack alerts; the base LLM classifier $f$ is not retrained or fine-tuned.

\begin{figure}[t]
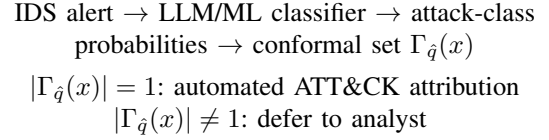

\centering
\fbox{\parbox{0.92\linewidth}{\centering
\vspace{0.8em}
IDS alert $\rightarrow$ LLM/ML classifier $\rightarrow$ attack-class probabilities $\rightarrow$ conformal set $\Gamma_{\hat{q}}(x)$\\[0.4em]
$|\Gamma_{\hat{q}}(x)|=1$: automated ATT\&CK attribution \qquad
$|\Gamma_{\hat{q}}(x)|\neq1$: defer to analyst
\vspace{0.8em}}}
\caption{CRC-based ATT\&CK alert triage framework. The abstention layer wraps the base classifier without modifying it. The system emits an automated ATT\&CK attribution only when the conformal set contains a single technique; empty or multi-technique sets are deferred.}
\label{fig:architecture}
\end{figure}

\begin{algorithm}[t]
\caption{CRC-based ATT\&CK Alert Triage}
\label{alg:crc}
\begin{algorithmic}[1]
\REQUIRE Calibration set $\mathcal{D}_{\text{cal}} = \{(x_i, y_i)\}_{i=1}^{n}$, target FAR $\alpha$, classifier $f$
\ENSURE Threshold $\hat{q}$
\STATE Compute NC scores: $s_i \gets s(x_i, y_i)$ for all $i \in [n]$ \hfill \textit{// Eq.~\eqref{eq:nc_score}}
\STATE Sort: $s_{(1)} \leq s_{(2)} \leq \cdots \leq s_{(n)}$
\STATE $\hat{q} \gets s_{(\lceil (1 - \alpha)(n + 1) \rceil)}$
\STATE \textbf{At test time}, for each alert $x$:
\STATE \quad $\Gamma_{\hat{q}}(x) \gets \{t_k \in \mathcal{T}: s(x,t_k) \leq \hat{q}\}$
\IF{$|\Gamma_{\hat{q}}(x)| = 1$}
    \STATE Output the unique technique in $\Gamma_{\hat{q}}(x)$
\ELSE
    \STATE Abstain: $\hat{t}(x) = \bot$ \hfill \textit{// Defer to analyst}
\ENDIF
\end{algorithmic}
\end{algorithm}

\subsection{Proxy Validity}
\label{sec:proxy}

Our nonconformity score $s(x, t_k)$ is derived from the LLM's softmax probabilities over IDS attack categories, not from a model trained directly on ATT\&CK labels.
The validity of this proxy rests on the bijective mapping $\phi$ (Section~\ref{sec:mapping}): since each attack category corresponds to exactly one ATT\&CK technique, the ordering induced by $f_k(x)$ over attack categories is identical to the ordering over ATT\&CK techniques.

The conformal risk statement does not require calibrated probabilities, but utility depends strongly on whether the score ranks plausible techniques ahead of implausible ones---this ranking quality is exactly what $\kappa(f)$ (Theorem~\ref{thm:singleton_capacity}) measures directly from $p_{(1)}(x)$ and $p_{(2)}(x)$, so we do not additionally require a separate calibration diagnostic (Section~\ref{sec:limitations} reports one supplementary check).

\subsection{Practical Considerations}

\paragraph{Calibration set size.}
The CRC threshold $\hat{q}$ depends on the calibration set size $n$.
Small $n$ yields conservative thresholds (high abstention rate), while large $n$ gives tighter thresholds.
We investigate the $n$--utility tradeoff in Section~\ref{sec:ablation}.

\paragraph{Model selection.}
CRC is a post-hoc wrapper: it does not constrain the choice of base classifier.
A model whose top-1 and top-2 attack-class probabilities are well-separated on most inputs (high singleton capacity $\kappa(f)$, Theorem~\ref{thm:singleton_capacity}) admits lower abstention rates for the same risk target, but the risk bound itself does not require this separation.
A model with poor separation simply triggers more abstentions, which is the correct conservative behavior---and, if $\kappa(f)$ is itself low, no calibration choice can recover a higher action rate (Section~\ref{sec:failure}).

\section{Experimental Setup}
\label{sec:setup}

\subsection{Datasets}

We use three IDS benchmark datasets spanning different network environments and attack types (Table~\ref{tab:datasets}).
Only attack-class samples are used (Normal traffic is excluded), since ATT\&CK attribution applies only to detected threats.

\begin{table*}[t]
\centering
\caption{Dataset characteristics. Only attack-class samples are used for ATT\&CK attribution.}
\label{tab:datasets}
\begin{tabular}{lccc}
\toprule
Dataset & Attack Classes & ATT\&CK Techniques & Test Samples \\
\midrule
CIC-IDS-2018~\cite{sharafaldin2018toward} & 3 & T1498, T1110, T1190 & 1{,}488 \\
HIKARI-2021~\cite{ferriyan2021generating} & 2 & T1110, T1190 & 226 \\
RT-IoT2022~\cite{sharmila2023detection} & 3 & T1498, T1190, T1046 & 1{,}779 \\
\bottomrule
\end{tabular}
\end{table*}

\subsection{LLM Backbones}

We evaluate six open-source LLMs spanning three model families and parameter scales:
Gemma-2 9B~\cite{team2024gemma}, LLaMA-3 8B~\cite{dubey2024llama}, Mistral 7B~\cite{jiang2023mistral}, Qwen-3 8B, Qwen-3 14B, and Qwen-3 32B~\cite{yang2025qwen3}.
All models perform zero-shot IDS classification: network flow features are serialized into natural-language prompts, and the model outputs logit-based probabilities over attack categories.
No fine-tuning is performed.

\subsection{ML Baselines}

To evaluate whether the CRC framework benefits from LLM backbones over traditional classifiers, we include two ML baselines with the same CRC abstention layer:
XGBoost~\cite{chen2016xgboost} and LightGBM~\cite{ke2017lightgbm}, trained on the same attack-class features with \texttt{predict\_proba} providing the softmax-equivalent probability vector.

\subsection{Baselines}

\begin{itemize}[leftmargin=*]
\item \textbf{B1-Argmax}: The LLM's top-1 prediction with no abstention. FAR equals the model's misclassification rate.
\item \textbf{B2-Threshold}: Abstain when $\max_k f_k(x) < \tau$ for a hand-tuned threshold $\tau$. We sweep $\tau$ to find the best FAR--utility tradeoff.
\item \textbf{B3-CRC} (ours): CRC-calibrated abstention with a formal false-attribution risk target.
\item \textbf{B4-ML-CRC}: XGBoost or LightGBM with the same CRC abstention framework, using identical NC score design and calibration procedure.
\end{itemize}

\subsection{Evaluation Metrics}

\begin{itemize}[leftmargin=*]
\item \textbf{FAR} (False Automated Attribution Risk): Fraction of all attack alerts that are automatically attributed to the wrong technique.
\item \textbf{AR} (Abstention Rate): Fraction of alerts deferred to human analysts.
\item \textbf{Utility}: Fraction of all attack alerts that receive a correct automated attribution.
\item \textbf{Automation rate}: $1-\text{AR}$, the fraction of alerts receiving any automated singleton attribution, whether correct or incorrect.
\item \textbf{Conditional error}: Fraction of automated singleton attributions that are wrong. This diagnostic metric is reported to characterize model behavior, but it is not the formal conformal risk target.
\item \textbf{Coverage}: Fraction of alerts where the true technique is in the prediction set (for multi-label evaluation).
\end{itemize}

\subsection{Configuration}

All LLM+CRC experiments use $n_{\text{cal}} = 200$ calibration samples (stratified by class), the remaining attack samples for testing, 5 random seeds for calibration/test splits, and $\alpha \in \{0.01, 0.05, 0.10, 0.20\}$.
The primary evaluation point is $\alpha = 0.05$.
All computations are CPU-only post-processing of pre-computed LLM logit outputs.

\section{Results}
\label{sec:results}

\subsection{RQ1: Does CRC Guarantee FAR $\leq \alpha$?}
\label{sec:rq1}

Table~\ref{tab:guarantee} summarizes the empirical false automated attribution risk across all configurations.
We apply a stringent operational criterion---FAR $\leq \alpha$ across \emph{all} 5 random seeds---which is stricter than the marginal risk statement itself.
Under this criterion, 65 of 72 configurations pass (90.3\%).

\begin{table*}[t]
\centering
\caption{CRC empirical verification across 18 model--dataset configurations and 4 $\alpha$ levels. FAR is the false automated attribution risk over all attack alerts. ``Pass'' indicates empirical FAR $\leq \alpha$ across all 5 random seeds. $\pm$ denotes standard deviation over seeds.}
\label{tab:guarantee}
\begin{tabular}{lcccc}
\toprule
$\alpha$ & FAR (mean $\pm$ std) & Utility (mean) & Coverage (mean) & Pass Rate \\
\midrule
0.01 & $0.002 \pm 0.003$ & 0.704 & 0.995 & 17/18 \\
0.05 & $0.015 \pm 0.014$ & 0.834 & 0.960 & 16/18 \\
0.10 & $0.025 \pm 0.031$ & 0.813 & 0.910 & 16/18 \\
0.20 & $0.040 \pm 0.058$ & 0.745 & 0.820 & 16/18 \\
\midrule
\multicolumn{4}{l}{Overall} & \textbf{65/72} \\
\bottomrule
\end{tabular}
\end{table*}

The 7 exceedances are concentrated in four model--dataset pairs: CIC-IDS-2018 $\times$ Mistral-7B (2 exceedances at $\alpha \in \{0.10, 0.20\}$), CIC-IDS-2018 $\times$ Qwen-3 14B (1 at $\alpha=0.05$), HIKARI-2021 $\times$ Gemma-2 9B (2 at $\alpha \in \{0.05, 0.20\}$), and RT-IoT2022 $\times$ Gemma-2 9B / Qwen-3 14B (2 combined).
In these cases the mean FAR is often below~$\alpha$ while at least one seed exceeds it, consistent with finite-sample variability under a stringent per-seed check.
We diagnose these exceedances directly in Section~\ref{sec:rq2.5} using the singleton-capacity framework of Section~\ref{sec:theory_contract}, which replaces an earlier nonconformity-score-variance heuristic that Proposition~\ref{prop:variance_counterexample} shows is not sufficient in general.

\subsection{RQ2: How Much Utility Does CRC Retain?}
\label{sec:rq2}

Table~\ref{tab:permodel} reports per-configuration results at $\alpha = 0.05$.

\begin{table*}[t]
\centering
\caption{Per-configuration FAR, Utility, and Coverage at $\alpha = 0.05$. Values are means $\pm$ std over 5 random seeds. Bold indicates Utility $\geq$ 0.90.}
\label{tab:permodel}
\begin{tabular}{llccc}
\toprule
Dataset & Model & FAR (mean $\pm$ std) & Utility & Coverage \\
\midrule
\multirow{6}{*}{CIC-IDS-2018}
& Gemma-2 9B & $0.035 \pm 0.002$ & 0.516 & 0.954 \\
& LLaMA-3 8B & $0.031 \pm 0.005$ & \textbf{0.961} & 0.961 \\
& Mistral 7B & $0.004 \pm 0.001$ & 0.378 & 0.996 \\
& Qwen-3 8B & $0.024 \pm 0.006$ & \textbf{0.952} & 0.952 \\
& Qwen-3 14B & $0.049 \pm 0.014$ & 0.838 & 0.950 \\
& Qwen-3 32B & $0.033 \pm 0.001$ & \textbf{0.954} & 0.954 \\
\midrule
\multirow{6}{*}{HIKARI-2021}
& Gemma-2 9B & $0.023 \pm 0.031$ & 0.831 & 0.977 \\
& LLaMA-3 8B & $0.000 \pm 0.000$ & \textbf{0.985} & 0.985 \\
& Mistral 7B & $0.000 \pm 0.000$ & \textbf{0.985} & 0.985 \\
& Qwen-3 8B & $0.000 \pm 0.000$ & \textbf{0.946} & 0.946 \\
& Qwen-3 14B & $0.000 \pm 0.000$ & \textbf{0.962} & 0.962 \\
& Qwen-3 32B & $0.000 \pm 0.000$ & \textbf{0.946} & 0.946 \\
\midrule
\multirow{6}{*}{RT-IoT2022}
& Gemma-2 9B & $0.036 \pm 0.004$ & 0.063 & 0.960 \\
& LLaMA-3 8B & $0.001 \pm 0.000$ & \textbf{0.959} & 0.959 \\
& Mistral 7B & $0.000 \pm 0.000$ & \textbf{0.954} & 0.954 \\
& Qwen-3 8B & $0.003 \pm 0.003$ & \textbf{0.959} & 0.959 \\
& Qwen-3 14B & $0.029 \pm 0.008$ & \textbf{0.904} & 0.951 \\
& Qwen-3 32B & $0.006 \pm 0.002$ & \textbf{0.927} & 0.952 \\
\bottomrule
\end{tabular}
\end{table*}

Across 18 configurations, the mean utility is 0.834, meaning that 83.4\% of attack alerts receive a correct automated technique attribution.
The mean automation rate, which also includes wrong singleton attributions, is 0.850.
Among the 12 configurations with utility above 0.90, FAR is simultaneously held below $\alpha = 0.05$, showing that high correct automation and low false-attribution risk are jointly achievable when the base classifier's score geometry favors singleton predictions.

The three low-utility outliers---CIC $\times$ Mistral (0.378), CIC $\times$ Gemma (0.516), and RT-IoT $\times$ Gemma (0.063)---correspond to models whose base classification accuracy or score separation is weak on those datasets.
CRC responds by increasing abstention, reducing the number of wrong automated attributions but also reducing usable automation.
This behavior is operationally conservative, but it is not a substitute for choosing a competent base model.

The best-performing configurations on HIKARI achieve 98.5\% utility with FAR $= 0$, so the calibrated layer abstains on only 1.5\% of alerts.
This suggests that the abstention layer has low overhead when the base model already separates the attack classes well.

\subsection{RQ2.5: Verifying the Decision-Contract Predictions}
\label{sec:rq2.5}

We test predictions P1--P7 of Section~\ref{sec:theory_predictions} against the existing per-sample logit data for all 18 LLM configurations, with no new inference.
For each configuration, $p_{(1)}(x)$ and $p_{(2)}(x)$ are recovered from the stored attack-normalized nonconformity scores, and $\hat\kappa$ is computed as the empirical CDF-gap supremum of Eq.~\eqref{eq:singleton_capacity}.
P1, P2, P4, and P7 check exact accounting identities (Propositions~\ref{prop:singleton_interval} and~\ref{prop:risk_factorization}, Theorems~\ref{thm:singleton_capacity}, \ref{thm:error_conservation}, and~\ref{thm:surrogate_slack}) that must hold by construction on any correctly computed data; we report them briefly, as an implementation check, before turning to the non-trivial empirical questions---P3, P5, and P6---that the identities alone do not answer: \emph{does} the capacity bound explain the observed low-utility configurations, \emph{how large} is semantic masking in practice, and \emph{does} a genuine structural-incapacity instance exist in this benchmark suite.

\paragraph{P1, P2, P4, P7 (accounting identities).}
All four hold with zero violations across the 18-configuration, 4-$\alpha$ matrix (difference $<10^{-9}$ throughout): the realized automation rate matches the CDF-gap prediction $\widehat F_2^-(\tau)-\widehat F_1^-(\tau)$ exactly (P1, Proposition~\ref{prop:singleton_interval}); $\mathrm{Utility}(\alpha)\leq\hat\kappa$ (P2, Theorem~\ref{thm:singleton_capacity}); $S(\tau)=C(\tau)-R(\tau)\geq0$ and is strictly positive wherever any deferral occurs, extending the single-configuration check of Table~\ref{tab:slack} (P4, Theorem~\ref{thm:surrogate_slack}); and, at the finest label resolution used throughout this paper where $M(C)\equiv0$, $B(h)=R(C)+D(C)$ (P7, the $\mathcal{E}=\mathcal{E}_{\mathrm{fine}}$ special case of Theorem~\ref{thm:error_conservation}, with the general $M(C)>0$ case reported as P5 below).

\paragraph{P5 (semantic masking under genuine coarsening).}
The 4-technique ATT\&CK mapping used elsewhere in this paper (Table~\ref{tab:mapping}) is bijective at both the technique and tactic level, so $M_\phi(g)\equiv 0$ in every experiment reported so far---the identity of Theorem~\ref{thm:coarsening_transfer} is verified only in its trivial boundary case up to this point.
To test the identity under a genuinely many-to-one map, we exploit real (non-synthetic) fine-grained attack subtypes present in the raw CIC-IDS-2018 labels before the dataset-specific loader collapses them into the four coarse categories: the DoS category aggregates 7 raw subtypes (e.g., DDoS-HOIC, DoS-Hulk, DoS-SlowHTTPTest), CredentialAccess aggregates 4 (e.g., FTP-BruteForce, SSH-Bruteforce), and Exploitation aggregates 3 (Bot, SQL Injection, Infilteration)---the coarse-to-fine map used by the benchmark curators themselves, not a construction of this paper.
We train an XGBoost classifier at this fine resolution (the sampled training data realizes 8 of the 14 raw attack subtypes, including a genuine 4-element DoS fiber: GoldenEye, Hulk, SlowHTTPTest, Slowloris), apply the singleton CRC rule to obtain a fixed selector $g$, and measure $R_{\mathrm{fine}}(g)$, $R_\phi(g)$, and $M_\phi(g)$ directly under the coarsening map $\phi$ back to the 4-category taxonomy, on 47{,}425 held-out attack alerts.
The identity $R_\phi(g) = R_{\mathrm{fine}}(g) - M_\phi(g)$ holds exactly at all 4 $\alpha$ levels (difference $<10^{-9}$ in every case), confirming Theorem~\ref{thm:coarsening_transfer} under a real, non-synthetic many-to-one map rather than only its trivial single-fiber boundary case.
The realized masking mass is small in this instance---$M_\phi(g)=2.1\times10^{-5}$ at $\alpha=0.05$ (the only $\alpha$ level with a non-zero fine-grained error among automated alerts) and $M_\phi(g)\approx 0$ elsewhere---because this classifier's rare fine-grained confusions occur predominantly \emph{across} coarse categories rather than within the same fiber.
This does not weaken the theorem: Theorem~\ref{thm:coarsening_transfer} is an exact accounting identity regardless of the masking magnitude, and Theorem~\ref{thm:reverse_impossibility} remains a worst-case existence result rather than a claim about typical deployments.
The practical reading is that coarsening-induced risk-hiding is a real, exactly quantifiable phenomenon whose magnitude is classifier- and taxonomy-dependent, and Theorem~\ref{thm:coarsening_transfer} gives the tool to measure it directly rather than assume it is negligible.

\paragraph{Finite-sample capacity certificate at readiness-sample size.}
Proposition~\ref{prop:capacity_certificate} is only useful if its DKW margin is tight enough, at a realistic unlabeled readiness-sample size, to distinguish structural incapacity from threshold misalignment before deployment---the question a pre-deployment practitioner actually faces, not whether the bound holds in the infinite-sample limit.
We test this directly: for each of the 18 configurations, we draw $m=200$ (matching the paper's calibration-set size, Section~\ref{sec:setup}) unlabeled subsamples from the full test population 500 times, compute $\hat\kappa$ on each subsample, and check whether it falls within the claimed margin $2\varepsilon_m(\delta)$ of the population $\kappa(f)$ at $\delta=0.1$.
Across all $18\times500=9{,}000$ trials, the margin holds in every case (pooled coverage $9{,}000/9{,}000$), against the $\geq90\%$ nominal target---the DKW bound is conservative here, as expected, but not so loose as to be uninformative: at $m=200$ the margin is $2\varepsilon_{200}(0.1)\approx0.192$, tight enough to certify structural incapacity ($\hat\kappa+2\varepsilon_m<\rho$) whenever the true $\kappa(f)$ is more than about $0.19$ below the action-rate floor $\rho$.
For the smallest test population in this study (HIKARI-2021, Table~\ref{tab:kappa}), the achievable margin is correspondingly wider, illustrating that the certificate's practical bite depends on the available readiness-sample size, not only on $\delta$.

\paragraph{P3 (structural incapacity, risk-constrained incapacity, or threshold misalignment).}
Table~\ref{tab:kappa} reports $\hat\kappa$, the risk-feasible capacity $\hat\kappa_\alpha$ (Definition~\ref{def:risk_feasible_capacity}, $\alpha=0.05$), and $A(\hat\tau)$, all computed on one fixed 200-sample calibration / held-out test split (seed $42$, matching Section~\ref{sec:setup}), with $\hat\kappa$ and $\hat\kappa_\alpha$ evaluated by exact enumeration of the singleton-interval breakpoints of Proposition~\ref{prop:singleton_interval} together with $\hat\tau$ itself, rather than a uniform grid: since $\hat\tau$ is always a feasible candidate in $\hat\kappa_\alpha$'s supremum when $R(\hat\tau)\leq\alpha$ on this split, this construction guarantees $\hat\kappa_\alpha\geq A(\hat\tau)$ by definition, with no discretization gap.
These single-split values differ slightly from Table~\ref{tab:permodel}'s 5-seed means (e.g.\ CIC $\times$ Gemma-2: $0.552$ here vs.\ $0.516$ there); we use one fixed split throughout this diagnostic so that $\hat\kappa$, $\hat\kappa_\alpha$, and $A(\hat\tau)$ are directly comparable to each other.
$\hat\kappa_\alpha$ itself is a post-hoc oracle diagnostic, not a deployable decision rule: it uses test labels to characterize \emph{whether} a better threshold exists at all, the same question Section~\ref{sec:rq3} correctly refuses to let B2-Threshold answer by tuning $\tau$ on held-out labels.  Concluding that a specific configuration is \emph{fixable in practice} additionally requires exhibiting a threshold selected without test labels; we do this separately for the one case where it matters, below.
Every configuration in the LLM matrix has $\hat\kappa\geq0.80$, ruling out structural incapacity at any deployment floor $\rho\leq0.80$, and in particular at the representative $\rho=0.5$ used throughout this section (P2); a stricter floor (e.g.\ $\rho=0.9$) would place RT-IoT $\times$ Gemma-2's $\hat\kappa=0.803$ into the structural-incapacity regime instead, since incapacity is a property of the pair $(\rho,f)$, not of $f$ alone.
But $\hat\kappa$ overstates what CRC could actually reach at $\alpha=0.05$: for CIC $\times$ Mistral, $\hat\kappa=1.000$ suggests substantial unused capacity, while $\hat\kappa_\alpha=0.380$ is almost exactly $A(\hat\tau)=0.380$---the calibrated threshold is already the best any threshold on this split can do without exceeding the risk target, and this is \emph{risk-constrained incapacity}, not a fixable misalignment.
RT-IoT $\times$ Gemma-2 shows the same pattern ($\hat\kappa=0.803$, $\hat\kappa_\alpha=0.122$, $A(\hat\tau)=0.111$).
CIC $\times$ Gemma-2 is different: the oracle $\hat\kappa_\alpha=0.637$ remains well above $A(\hat\tau)=0.552$, an $8.5$-point gap suggesting the shortfall may be recoverable.
We verify this without appeal to test labels: searching Proposition~\ref{prop:singleton_interval}'s singleton intervals for a threshold satisfying $R(\tau)\leq\alpha$ using only the 200-sample calibration split (the same split that produced $\hat\tau$) yields an alternative threshold with calibration-split action rate $0.585$; freezing this threshold and evaluating it on the disjoint held-out test split gives $A=0.605$ (vs.\ CRC's $0.552$ on the same split) at empirical risk $0.036\leq\alpha$---a genuine, label-free-at-selection-time improvement, confirming recoverable threshold misalignment for this configuration specifically rather than inferring it from the oracle $\hat\kappa_\alpha$ alone.
This distinction matters because $\hat\kappa$ and the actionability gap $\mathrm{AG}=\hat\kappa-A(\hat\tau)$, while still a valid label-free necessary condition and upper bound (Theorem~\ref{thm:singleton_capacity}), are not sufficient on their own to diagnose which configurations are recoverable: two of the three cases with the largest $\mathrm{AG}$ turn out not to be, and only a labeled check of $\hat\kappa_\alpha$ against $A(\hat\tau)$ resolves it.
This is still a strictly stronger and falsifiable replacement for the retracted NC-variance heuristic (Proposition~\ref{prop:variance_counterexample}), which offered no such distinction at all; Section~\ref{sec:exp-degenerate} additionally tests whether genuine structural incapacity ($\hat\kappa$ near zero) occurs anywhere in the ML+CRC baselines.

\begin{table*}[t]
\centering
\caption{Singleton capacity $\hat\kappa$, risk-feasible capacity $\hat\kappa_\alpha$, and realized action rate at $\alpha=0.05$ for the three lowest-utility LLM configurations, plus the range across the remaining 15, all on one fixed calibration/test split (seed 42) with $\hat\kappa,\hat\kappa_\alpha$ computed by exact breakpoint enumeration so that $\hat\kappa_\alpha\geq A(\hat\tau)$ holds by construction. $\hat\kappa$ alone would misclassify two of the three as recoverable threshold misalignment; $\hat\kappa_\alpha$ shows only CIC~$\times$~Gemma-2 actually is.}
\label{tab:kappa}
\begin{tabular}{llcccl}
\toprule
Dataset & Model & $\hat\kappa$ & $\hat\kappa_\alpha$ & $A(\hat\tau)$ & Regime \\
\midrule
RT-IoT2022 & Gemma-2 9B & 0.803 & 0.122 & 0.111 & Risk-constrained incapacity \\
CIC-IDS-2018 & Mistral 7B & 1.000 & 0.380 & 0.380 & Risk-constrained incapacity \\
CIC-IDS-2018 & Gemma-2 9B & 0.915 & 0.637 & 0.552 & Threshold misalignment \\
\multicolumn{6}{l}{\emph{(remaining 15 configurations: $\hat\kappa_\alpha \geq 0.923$, within $0.08$ of $A(\hat\tau)$ in all)}} \\
\bottomrule
\end{tabular}
\end{table*}

\phantomsection\label{sec:exp-certificate}
\paragraph{P6 (finite-sample actionability certificate).}
We apply Proposition~\ref{prop:action_rate_certificate} rather than the raw empirical action rate: $\hat\tau$ is fixed on a 200-sample calibration split, and $\underline A_\delta$ is computed at $\delta=0.1$ from the disjoint certification split (the remaining test samples; $m=1{,}288$ for CIC-IDS-2018, $1{,}579$ for RT-IoT2022, and $26$ for HIKARI-2021, whose small test set gives a correspondingly wide margin $\varepsilon_m\approx0.21$).
Certifying $(\alpha,\rho)$-actionability at $\alpha=0.05,\rho=0.5$ against $\underline A_\delta$ (not the point estimate) still passes 16 of 18 configurations, and the two that fail---CIC $\times$ Mistral ($\underline A_\delta=0.350$) and RT-IoT $\times$ Gemma-2 ($\underline A_\delta=0.084$)---are exactly the two Table~\ref{tab:kappa} configurations whose risk-feasible capacity $\hat\kappa_\alpha$ is itself below $\rho=0.5$ (Table~\ref{tab:kappa}'s risk-constrained-incapacity cases): no calibration choice at $\alpha=0.05$ could have certified these regardless of sample size, which is a stronger and more specific statement than a failed empirical check.
For the 16 certified configurations, Theorem~\ref{thm:actionability_bridge} with $\underline A_\delta$ in place of $\rho$ gives $U \geq \underline A_\delta - 0.05$ and $Q \leq 0.05/\underline A_\delta$ on the $0.9$-probability event that the action-rate bound holds, combined with the risk side's standard marginal conformal guarantee (Theorem~\ref{thm:crc})---a stated confidence level on the action-rate side, rather than an unqualified point estimate, though not a single joint confidence level across both sides (Proposition~\ref{prop:action_rate_certificate}).

\subsection{RQ3: CRC vs.\ Baselines}
\label{sec:rq3}

\paragraph{CRC vs.\ Argmax across all configurations.}
B1-Argmax (no abstention) uses the LLM's top-1 prediction directly, so FAR equals the base misclassification risk regardless of $\alpha$.
Across all 18 configurations at $\alpha = 0.05$, B1-Argmax violates the FAR target in 6 of 18 cases (where the base error rate exceeds 0.05), while B3-CRC satisfies it in 16 of 18.
In the 12 configurations where B1-Argmax happens to satisfy $\alpha = 0.05$ (because the model is already accurate enough), CRC achieves comparable utility (mean 0.94 vs.\ B1's 0.97) while adapting the decision rule to the target risk level.

\paragraph{Detailed three-way comparison.}
Table~\ref{tab:baselines} compares all three baselines on CIC-IDS-2018 $\times$ Qwen-3 8B, where the base error rate (0.033) makes the FAR--utility tradeoff non-trivial.
B2-Threshold here acts on the predicted-class confidence $p_{(1)}(x)$ only (Definition of $g_\tau$, Section~\ref{sec:singleton_geometry}): the threshold $\tau$ is swept on the calibration split to find the smallest $\tau$ (highest action rate) whose \emph{calibration-set} FAR satisfies the target, then frozen and applied to the test split.  An earlier version of this baseline selected $\tau$ using the true test label directly, which is not a deployable decision rule; we report the corrected, label-free version throughout.

\begin{table}[t]
\centering
\caption{Baseline comparison on CIC-IDS-2018 $\times$ Qwen-3 8B. Utility is the correct automated attribution rate. B2-Threshold's $\tau$ is calibrated on the cal split only and frozen before evaluation on test.}
\label{tab:baselines}
\begin{tabular}{llccc}
\toprule
$\alpha$ & Method & FAR & Utility & Coverage \\
\midrule
\multirow{3}{*}{0.01}
& B1-Argmax & 0.033 & 0.967 & --- \\
& B2-Threshold & 0.001 & 0.858 & --- \\
& B3-CRC & 0.001 & 0.858 & 0.999 \\
\midrule
\multirow{3}{*}{0.05}
& B1-Argmax & 0.033 & 0.967 & --- \\
& B2-Threshold & 0.033 & 1.000 & --- \\
& B3-CRC & 0.030 & 0.959 & 0.959 \\
\midrule
\multirow{3}{*}{0.10}
& B1-Argmax & 0.033 & 0.967 & --- \\
& B2-Threshold & 0.033 & 1.000 & --- \\
& B3-CRC & 0.016 & 0.933 & 0.933 \\
\midrule
\multirow{3}{*}{0.20}
& B1-Argmax & 0.033 & 0.967 & --- \\
& B2-Threshold & 0.033 & 1.000 & --- \\
& B3-CRC & 0.000 & 0.780 & 0.780 \\
\bottomrule
\end{tabular}
\end{table}

At $\alpha=0.01$, B2 and B3 are statistically indistinguishable ($\tau=0.839$ vs.\ $1-\hat q=0.836$): both restrict action to a comparably high-confidence region and reach nearly identical FAR and utility.
At $\alpha \in \{0.05, 0.10, 0.20\}$, the cal-split FAR of the unrestricted argmax rule (0.033) already satisfies the target, so the calibrated $\tau$ collapses to $0$ and B2 degenerates to B1-Argmax exactly---both act on every alert, and both happen to realize test-set FAR below $\alpha$ in this particular split.
On this single configuration, CRC does not dominate a correctly calibrated global threshold in raw utility; the earlier reported 25.7\% utility gap at $\alpha=0.05$ was an artifact of a baseline that selected $\tau$ using held-out true labels and is not reproduced here.
The distinction that survives is the one identified in Section~\ref{sec:coarsening_theory} and Section~\ref{sec:singleton_geometry}: CRC's threshold $\hat q$ is the calibration order statistic of Theorem~\ref{thm:crc} and inherits its finite-sample, distribution-free coverage guarantee under exchangeability, whereas B2's $\tau$ is selected by unconstrained grid search against the calibration-set empirical FAR and carries no comparable guarantee---it can overfit a small or unlucky calibration split, a risk that does not appear in this single low-noise configuration but is not bounded by any theorem for B2.
We quantify the operational cost of using CRC's conformal miscoverage as a surrogate for wrong-automation risk next.

\phantomsection\label{sec:exp-slack}
\paragraph{Surrogate slack.}
Theorem~\ref{thm:surrogate_slack} decomposes conformal miscoverage $C(\tau)=\mathbb{P}(y\notin\Gamma_\tau(x))$ into the operational risk $R(\tau)$ actually incurred by automated decisions and the deferred-miscoverage slack $S(\tau)=C(\tau)-R(\tau)$ spent on alerts that are abstained anyway.
Table~\ref{tab:slack} reports this decomposition for B3-CRC on CIC-IDS-2018 $\times$ Qwen-3 8B, using $C(\tau)=1-\text{coverage}$ from Table~\ref{tab:baselines}.

\begin{table}[t]
\centering
\caption{Surrogate-slack decomposition (Theorem~\ref{thm:surrogate_slack}) for B3-CRC on CIC-IDS-2018 $\times$ Qwen-3 8B. $S(\tau)=C(\tau)-R(\tau)\geq 0$ in every row, confirming the identity; $S$ grows with $\alpha$ as more of the miscoverage budget is spent on already-deferred alerts.}
\label{tab:slack}
\begin{tabular}{rcccc}
\toprule
$\alpha$ & $C(\tau)$ & $R(\tau)$ & $S(\tau)$ & AR \\
\midrule
0.01 & 0.0008 & 0.0008 & 0.0000 & 0.141 \\
0.05 & 0.0411 & 0.0295 & 0.0116 & 0.012 \\
0.10 & 0.0668 & 0.0163 & 0.0505 & 0.050 \\
0.20 & 0.2197 & 0.0000 & 0.2197 & 0.220 \\
\bottomrule
\end{tabular}
\end{table}

At $\alpha=0.20$, the identity is most visible: all of the nominal $0.22$ miscoverage budget is spent on samples that are already abstained ($S=0.2197$, $R=0$), so the coverage-calibrated threshold is fully conservative relative to the operational objective at this operating point---conformal miscoverage control is valid here but tells us nothing about wrong-automation risk, which is exactly zero.
At $\alpha=0.05$, roughly $28\%$ of the miscoverage budget ($0.0116/0.0411$) is similarly non-operational.
This slack is the quantitative reason a practitioner should report $R(\tau)$ (or FAR) directly rather than treating conformal coverage as self-evidently the quantity of interest.

\paragraph{ML+CRC baseline (B4).}
To isolate the contribution of LLM backbones from the CRC framework itself, we apply the same CRC abstention layer to XGBoost and LightGBM classifiers trained on the same attack-class features.
Table~\ref{tab:ml_crc} reports the results at $\alpha = 0.05$.

\begin{table*}[t]
\centering
\caption{ML+CRC baseline (B4) vs.\ best LLM+CRC (B3) at $\alpha = 0.05$, after the training-stability correction of Section~\ref{sec:exp-degenerate}. Utility is the correct automated attribution rate.}
\label{tab:ml_crc}
\begin{tabular}{llccc}
\toprule
Dataset & Model & FAR (mean $\pm$ std) & Utility & Guarantee \\
\midrule
\multirow{3}{*}{CIC-IDS-2018}
& XGBoost+CRC & $0.000 \pm 0.000$ & 0.961 & \checkmark \\
& LightGBM+CRC & $0.000 \pm 0.000$ & 0.959 & \checkmark \\
& Best LLM+CRC & $0.031 \pm 0.005$ & \textbf{0.961} & \checkmark \\
\midrule
\multirow{3}{*}{HIKARI-2021}
& XGBoost+CRC & $0.000 \pm 0.000$ & 0.956 & \checkmark \\
& LightGBM+CRC & $0.000 \pm 0.000$ & 0.960 & \checkmark \\
& Best LLM+CRC & $0.000 \pm 0.000$ & \textbf{0.985} & \checkmark \\
\midrule
\multirow{3}{*}{RT-IoT2022}
& XGBoost+CRC & $0.000 \pm 0.000$ & 0.952 & \checkmark \\
& LightGBM+CRC & $0.0001 \pm 0.0000$ & 0.953 & \checkmark \\
& Best LLM+CRC & $0.001 \pm 0.000$ & \textbf{0.959} & \checkmark \\
\bottomrule
\end{tabular}
\end{table*}

On all three datasets, ML+CRC is competitive with the best LLM+CRC (utility within 1--3 percentage points), indicating that CRC's value is largely orthogonal to backbone choice when the base classifier's probability estimates are usable.
The largest gap is on HIKARI-2021, where the best LLM (LLaMA-3) retains a 2.5-point advantage over LightGBM (0.985 vs.\ 0.960).

This result clarifies the contribution of the backbone: the CRC framework is model-agnostic and delivers comparable utility with ML classifiers when their probability estimates are useful.
The observed advantage of the LLM backbones is consistency across these three datasets, not uniform dominance over every ML classifier.

\phantomsection\label{sec:exp-degenerate}
\paragraph{A retracted structural-incapacity claim, and what replaced it.}
An earlier version of this analysis reported that LightGBM+CRC on RT-IoT2022 abstains on 100\% of samples ($\text{FAR}=0$, $\text{Utility}=0$) at $\alpha \leq 0.10$, and interpreted this as a structural-incapacity instance under Theorem~\ref{thm:singleton_capacity}---the regime $\hat\kappa<\rho$, distinct from both risk-constrained incapacity and threshold misalignment in Table~\ref{tab:kappa}, neither of which requires $\hat\kappa$ itself to be low.
This did not survive a training-stability check.
Retraining LightGBM on the identical data, feature pipeline, and calibration procedure with 20 different random seeds ($\{0,\ldots,19\}$, holding the CRC calibration seeds fixed) produced a degenerate ($\text{Utility} < 0.1$) model in 0 of 20 runs; every retrained model achieves $\hat\kappa = 1.000$ and $\text{Utility} \approx 0.948$--$0.950$, matching XGBoost.
The original collapse is therefore best explained as a one-off training artifact---plausibly floating-point nondeterminism in multi-threaded histogram construction producing, on that single run, a true-label probability distribution with a rare but non-negligible high-nonconformity cluster large enough to saturate the conformal calibration quantile at $\hat q=1$---rather than a reproducible property of LightGBM or of this dataset.
We report this as a negative result: \emph{no structural-incapacity instance occurs anywhere in the 24 LLM+CRC and ML+CRC configurations evaluated in this paper.}
Theorem~\ref{thm:singleton_capacity} and Proposition~\ref{prop:capacity_certificate} remain applicable as a pre-deployment diagnostic for the regime, and the retracted case is retained here as a caution against drawing structural conclusions from a single training run of a stochastic learner, a risk the label-free $\hat\kappa$ certificate does not remove by itself---it must be paired with a stability check across retrainings when the base learner's fit is not itself guaranteed to be stable.

\paragraph{Pareto frontier analysis.}
Figure~\ref{fig:pareto} traces the FAR--Utility tradeoff for three representative models (Qwen-3 8B, LLaMA-3 8B, Qwen-3 32B) on CIC-IDS-2018 by sweeping $\alpha \in [0.001, 0.20]$.
All three models exhibit a sharp elbow between $\alpha = 0.01$ and $\alpha = 0.05$: utility jumps from $\sim$0.62--0.86 to $\sim$0.96, after which further relaxation of $\alpha$ yields diminishing returns.
LLaMA-3 8B achieves the earliest saturation (utility $> 0.96$ at $\alpha = 0.05$), while Qwen-3 32B exhibits a delayed elbow due to its bimodal NC score distribution (nc\_var $= 0.032$).

\begin{figure}[t]
\centering
\includegraphics[width=\columnwidth]{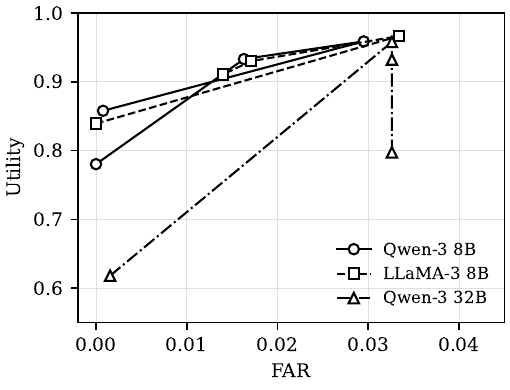}
\caption{FAR--Utility tradeoff on CIC-IDS-2018 for three representative LLMs, swept over $\alpha \in [0.001, 0.20]$. Markers trace the empirical operating points; the elbow near $\alpha=0.01$--$0.05$ marks the practical transition from strict risk control to the main operating point.}
\label{fig:pareto}
\end{figure}

\subsection{RQ4: ATT\&CK Hierarchy Consistency}
\label{sec:rq4}

Because $\phi$ is bijective, technique-to-tactic propagation of the CRC guarantee is a mathematical consequence of the label structure (Section~\ref{sec:formulation}), not an empirical finding; as a sanity check, FAR, utility, and coverage are numerically identical at both levels across all seeds and $\alpha$ (e.g., seed 42, $\alpha=0.01$: FAR$=0.00078$, utility$=0.858$, coverage$=0.999$ at both).
Under a genuinely many-to-one mapping the guarantee would still propagate but would no longer be tight, since some technique-level errors become correct at the tactic level---exactly the coarsening identity of Theorem~\ref{thm:coarsening_transfer}, verified directly under a real many-to-one map in Section~\ref{sec:rq2.5}.

\subsection{RQ5: Cross-Dataset Robustness}
\label{sec:rq5}

Table~\ref{tab:cross} tests whether a CRC threshold calibrated on one dataset transfers to another, for Qwen-3 8B and LLaMA-3 8B at $\alpha = 0.05$.

\begin{table}[t]
\centering
\caption{Cross-dataset generalization at $\alpha = 0.05$. ``Cal'' = calibration source, ``Test'' = test target. FAR remains below $\alpha$ in these 8 configurations, but utility varies widely.}
\label{tab:cross}
\begin{tabular}{llllccc}
\toprule
Model & Cal & Test & FAR & Utility & Guarantee \\
\midrule
\multirow{4}{*}{Qwen-3 8B}
& CIC & HIKARI & 0.000 & 1.000 & \checkmark \\
& CIC & RT-IoT & 0.011 & 0.962 & \checkmark \\
& HIKARI & CIC & 0.000 & 0.043 & \checkmark \\
& RT-IoT & CIC & 0.000 & 0.836 & \checkmark \\
\midrule
\multirow{4}{*}{LLaMA-3 8B}
& CIC & HIKARI & 0.000 & 1.000 & \checkmark \\
& CIC & RT-IoT & 0.003 & 0.992 & \checkmark \\
& HIKARI & CIC & 0.000 & 0.611 & \checkmark \\
& RT-IoT & CIC & 0.000 & 0.580 & \checkmark \\
\bottomrule
\end{tabular}
\end{table}

Empirical FAR remains below $\alpha$ in all 8 configurations despite exchangeability not being guaranteed across datasets---a robustness observation, not a formal transfer guarantee---but utility degrades substantially in several cases, down to 4.3\% (95.7\% deferred) for HIKARI$\to$CIC with Qwen-3.
The asymmetry is informative: CIC$\to$HIKARI transfers well (utility 1.0) because CIC's richer label space produces a threshold generous enough for HIKARI's easier 2-class problem, while HIKARI$\to$CIC fails because HIKARI's 2-class calibration is overly conservative for CIC's 3-class task.
In-domain calibration is therefore essential for deployment; adaptive or non-exchangeable conformal methods~\cite{gibbs2021adaptive,farinhas2024nonexchangeable} are natural extensions.

\subsection{RQ6: Calibration Set Size Ablation}
\label{sec:ablation}

Table~\ref{tab:ncal} shows the effect of calibration set size on FAR and utility (CIC-IDS-2018 $\times$ Qwen-3 8B, $\alpha = 0.05$).

\begin{table}[t]
\centering
\caption{Effect of calibration set size $n_{\text{cal}}$ on FAR and Utility (CIC-IDS-2018, Qwen-3 8B, $\alpha = 0.05$). Utility saturates at $n_{\text{cal}} \approx 200$; intermediate values (400, 600) fall within $0.01$ of the 800 endpoint.}
\label{tab:ncal}
\begin{tabular}{rccc}
\toprule
$n_{\text{cal}}$ & FAR & AR & Utility \\
\midrule
50 & 0.008 & 0.097 & 0.895 \\
100 & 0.024 & 0.022 & 0.954 \\
200 & 0.024 & 0.024 & 0.952 \\
800 & 0.024 & 0.017 & 0.959 \\
\bottomrule
\end{tabular}
\end{table}

Utility rises sharply from $n_{\text{cal}}=50$ (0.895) to $100$ (0.954), then saturates within $0.01$ through $800$, while FAR stays well below $\alpha=0.05$ throughout; at $n_{\text{cal}}=50$ the threshold is overly conservative (FAR$=0.008\ll0.05$), a finite-sample artifact where an imprecise $\hat q$ makes CRC default to caution.
Consequently, 200 labeled attack samples suffice for near-optimal CRC deployment---a modest requirement for most SOC environments.

\section{Discussion}
\label{sec:discussion}

\subsection{Failure Case Analysis}
\label{sec:failure}

An earlier version of this analysis diagnosed low-utility configurations using nonconformity-score variance combined with base accuracy.
Proposition~\ref{prop:variance_counterexample} shows this heuristic is not sound in general: two classifiers can share the same true-label score mean, variance, and top-one accuracy while having opposite singleton capacity ($\kappa=1$ vs.\ $\kappa=0$).
We replace it with the label-free diagnostic developed in Section~\ref{sec:singleton_geometry} and verified in Section~\ref{sec:rq2.5}.

Every low-utility configuration in the LLM matrix (Table~\ref{tab:kappa}) has high singleton capacity ($\hat\kappa\geq0.80$): RT-IoT $\times$ Gemma-2 ($\hat\kappa=0.803$), CIC $\times$ Mistral ($\hat\kappa=1.000$), and CIC $\times$ Gemma-2 ($\hat\kappa=0.915$).
By Theorem~\ref{thm:singleton_capacity}, this rules out structural incapacity for all three at $\rho=0.5$: the base classifier's top-1 and top-2 attack probabilities are separable enough, on most inputs, to support a much higher \emph{unconstrained} automation rate than CRC's calibrated $\hat\tau$ achieves.
It does not follow that this unused capacity is reachable under the risk target, and here the three cases diverge once $\hat\kappa_\alpha$ (Definition~\ref{def:risk_feasible_capacity}) is checked directly against the true labels.
For CIC $\times$ Mistral and RT-IoT $\times$ Gemma-2, $\hat\kappa_\alpha$ (0.380 and 0.122) sits almost exactly at the realized $A(\hat\tau)$ (0.380 and 0.111): CRC's calibrated threshold is already the best any threshold can do without exceeding $\alpha=0.05$ on this split, and the large unconstrained $\hat\kappa$ was, in these two cases, not informative about what recalibration could achieve---these are \emph{risk-constrained incapacity}, and the correct remedy is a less stringent $\alpha$ or a more competent base classifier, not a different threshold at the same $\alpha$.
CIC $\times$ Gemma-2 is the one case where the oracle $\hat\kappa_\alpha=0.637$ remains well above $A(\hat\tau)=0.552$, and here the recoverability is confirmed without test-label selection (Section~\ref{sec:rq2.5}): a threshold search restricted to the calibration split alone finds an alternative $\hat\tau$ that reaches $A=0.605$ at empirical risk $0.036\leq\alpha$ on the disjoint held-out test split, genuinely recoverable threshold misalignment for this configuration specifically.

The complementary failure mode---genuine structural incapacity, where $\hat\kappa$ itself is near zero and no threshold can raise the action rate at any risk level---is not observed anywhere in this study; Section~\ref{sec:exp-degenerate} reports an ML+CRC configuration that initially appeared to be such a case but did not survive a 20-seed training-stability check, and we retract that reading rather than retrofit a diagnosis to an artifact.
The practical pre-deployment check that follows from Theorem~\ref{thm:singleton_capacity} and Proposition~\ref{prop:capacity_certificate} is therefore staged, not a single number: estimate $\hat\kappa$ on an unlabeled readiness sample before calibration, and re-estimate it across a small number of retrainings when the base learner's fit is not itself guaranteed to be stable.
If $\hat\kappa$ is already below the required action-rate floor $\rho$, no amount of recalibration, larger $n_{\mathrm{cal}}$, or $\alpha$ relaxation will make the configuration deployable, and the correct remedy is to change the base classifier.
If $\hat\kappa\geq\rho$, this only rules out the geometric obstruction; a shortfall at $A(\hat\tau)<\rho$ still requires checking $\hat\kappa_\alpha$ (which needs labels) before concluding that recalibration---rather than a different $\alpha$ or a different classifier---is the right remedy.

\subsection{Practical Implications}

Section~\ref{sec:rq1}'s FAR check alone is not a deployment criterion: Proposition~\ref{prop:risk_factorization} and the retracted case in Section~\ref{sec:exp-degenerate} show that a low FAR is equally consistent with a well-functioning tiered workflow and with a system that automates almost nothing.
The framework instead supports a four-stage deployment procedure, each stage gating the next:

\begin{enumerate}[leftmargin=*]
\item \textbf{Capacity check}: before calibration, estimate $\hat\kappa$ on an unlabeled readiness sample (Proposition~\ref{prop:capacity_certificate}).
If $\hat\kappa$ falls below the required action-rate floor $\rho$, no calibration choice can reach $\rho$ (Theorem~\ref{thm:singleton_capacity}); replace the base classifier rather than proceeding.

\item \textbf{Calibration}: fix $\hat\tau$ on a labeled calibration split via Algorithm~\ref{alg:crc}, certifying $R(C_{\hat\tau})\leq\alpha$ (Theorem~\ref{thm:crc}).

\item \textbf{Actionability certification}: on a disjoint certification split, compute the action-rate lower bound $\underline A_\delta$ (Proposition~\ref{prop:action_rate_certificate}).
Only if $\underline A_\delta\geq\rho$ does the contract satisfy $(\alpha,\rho)$-actionability with confidence $1-\delta$ (Theorem~\ref{thm:actionability_bridge}); if it does not, check $\hat\kappa_\alpha$ against the same certification data (Definition~\ref{def:risk_feasible_capacity}) before choosing a remedy: $\hat\kappa_\alpha\geq\rho$ means a different threshold at the same $\alpha$ can recover the shortfall (recalibration), while $\hat\kappa_\alpha<\rho$ means the current $\alpha$ itself is the binding constraint, and only relaxing $\alpha$ or replacing the base classifier will help (Section~\ref{sec:failure}).

\item \textbf{Deployment}: alerts with $g_{\hat\tau}(x)=1$ are logged with the singleton ATT\&CK technique and routed to automated playbooks; alerts with $g_{\hat\tau}(x)=0$ are forwarded to analysts with the top-$k$ candidate techniques and their scores as decision support.
\end{enumerate}

Skipping stages 1 and 3 and deploying on the FAR check alone is exactly the failure mode this paper opened with: RT-IoT2022 $\times$ Gemma-2 (Table~\ref{tab:kappa}) passes stage 2 with $\text{FAR}=0.036\pm0.004\leq\alpha$ but has $\underline A_\delta=0.084$, and would clear a FAR-only gate while automating under 10\% of alerts.
In the configurations that pass all four stages, the workflow handles roughly 90--95\% of attack alerts correctly and automatically with the certified error bounds of Theorem~\ref{thm:actionability_bridge}.

\subsection{NC Score Design}

The ablation in E7 compares normalized NC scores (Eq.~\eqref{eq:nc_score}) against raw softmax probabilities.
At $\alpha = 0.05$, the normalized design achieves FAR $= 0.030$ with utility $= 0.959$, while the raw design yields FAR $= 0.032$ with utility $= 0.976$ but higher FAR variance across seeds.
Normalization produces tighter FAR control at a modest utility cost, which is the preferred tradeoff for safety-critical applications.

\subsection{Limitations}
\label{sec:limitations}

\paragraph{Label granularity and score ranking.}
Our LLM+CRC evaluation uses 3--4 ATT\&CK techniques per dataset, whereas the full ATT\&CK matrix contains over 600 techniques.
The CRC framework itself is agnostic to the number of classes, but $\kappa(f)$ (Theorem~\ref{thm:singleton_capacity}) may degrade with finer-grained technique taxonomies where LLM confusion between similar techniques increases; as a supplementary check on the current 4-technique mapping, the Spearman rank correlation between $f_k(x)$ and the ground-truth indicator averages $\rho=0.672$ (std $0.210$, minimum $-0.188$) across 48 model--dataset--class configurations, with $43.75\%$ of configurations exceeding $\rho=0.8$---consistent with $\kappa(f)\geq0.80$ holding everywhere in Table~\ref{tab:kappa} but not a guarantee that it will continue to at finer granularities.
Section~\ref{sec:coarsening_theory} shows that real fine-grained subtype labels exist in at least one benchmark (CIC-IDS-2018's raw attack labels, pre-dating the dataset's own 4-category collapse) and uses them to test the coarsening-transfer identity; extending the full LLM+CRC pipeline to this or a finer ATT\&CK sub-technique resolution, rather than the single-classifier XGBoost instantiation used there, is left for future work.

\paragraph{Input modality.}
Our system operates on structured network flow features, not unstructured text.
It is not comparable to CTI text-to-ATT\&CK systems (rcATT, TRAM, LADDER), which solve a fundamentally different task.
The two approaches are complementary: ours handles real-time IDS alerts, while CTI mappers process post-incident reports.

\paragraph{Model dependence.}
CRC utility is bounded by the base model's classification accuracy and score separation.
As shown in Table~\ref{tab:permodel}, poorly performing models (Gemma-2 on RT-IoT) result in near-total abstention.
A single training run of the ML+CRC baseline initially appeared to reinforce this with a complete-abstention case (LightGBM on RT-IoT, utility $=0$); Section~\ref{sec:exp-degenerate} shows this does not reproduce across 20 retrainings and should not be read as a stable property of either the classifier or the framework.
CRC does not improve model accuracy; it provides a principled mechanism to expose and manage model uncertainty.
In these experiments, LLM backbones offer more stable probability rankings across the three traffic distributions, even though ML classifiers are competitive on individual datasets.

\paragraph{Conditional error among automated outputs.}
The controlled FAR is an unconditional risk over all attack alerts.
It should not be read as the conditional error rate among the subset of alerts that receive automated attributions.
Remote recomputation from the raw logprob files shows that, at $\alpha=0.05$, the mean conditional automated-attribution error is 3.6\%, but RT-IoT2022 $\times$ Gemma-2 reaches 36.3\% because the model automates only 9.9\% of alerts.
This is why conditional error should be monitored as a deployment diagnostic even though the formal risk target is unconditional.

\paragraph{Exchangeability assumption.}
The conformal guarantee requires calibration and test samples to be exchangeable.
In practice, network traffic is non-stationary (new attack types emerge, traffic patterns shift).
Our cross-dataset experiments (Section~\ref{sec:rq5}) show low empirical FAR under the tested shifts, but they do not provide a formal guarantee once exchangeability is broken.
Utility degrades substantially under several shifts.
Adaptive conformal methods~\cite{gibbs2021adaptive,zaffran2022adaptive} that reweight or refresh calibration data over time are a natural extension for production SOC deployments, as are non-exchangeable risk-control formulations~\cite{farinhas2024nonexchangeable} that relax the exchangeability requirement directly rather than tracking a drifting threshold; either would extend Theorem~\ref{thm:crc}'s guarantee without changing the decision-contract accounting built on top of it.

\section{Conclusion}
\label{sec:conclusion}

This paper began from an operational puzzle---a calibrated conformal risk bound satisfied by a classifier that automates nothing---and generalized it: a risk certificate is a property of a decision contract, not of a predictor alone, and can be satisfied by weakening that contract through abstention or semantic coarsening rather than genuine predictive competence.
We formalized this as an error-conservation law ($B(h)=R(C)+D(C)+M(C)$, Theorem~\ref{thm:error_conservation}), an exact geometric characterization of singleton automation and the resulting capacity limit $\kappa(f)$ (Theorem~\ref{thm:singleton_capacity}), and a non-degenerate $(\alpha,\rho)$-actionable certificate whose action-rate side carries an explicit finite-sample lower confidence bound (Proposition~\ref{prop:action_rate_certificate}), distinct in kind from the risk side's standard conformal guarantee rather than resting on a single realized measurement.

Instantiated on ATT\&CK-aligned alert triage for LLM-based intrusion detection (3 datasets, 6 LLMs, 4 error-rate thresholds), the theory separates three regimes that a scalar risk number cannot distinguish.
The singleton-capacity diagnostic rules out structural incapacity for all three low-utility LLM configurations at deployment floor $\rho=0.5$, correcting what a variance-based heuristic would have wrongly permitted (Proposition~\ref{prop:variance_counterexample}); but $\kappa(f)$ alone is not sufficient to call the shortfall recoverable---the risk-feasible capacity $\kappa_\alpha(f)$ shows two of the three are risk-constrained incapacity, not threshold misalignment, and only one is genuinely fixable, verified by exhibiting the alternative threshold rather than inferred from $\kappa(f)$'s unconstrained bound.
A 20-seed training-stability re-run of the ML+CRC baseline shows an initially observed zero-utility collapse (LightGBM on RT-IoT2022) does not reproduce, so this study finds no confirmed instance of the third regime, genuine structural incapacity; we report this as an open question rather than overstate a single-run anomaly.
Using real fine-grained attack-subtype labels from CIC-IDS-2018's own data pipeline, we confirmed the coarsening-transfer identity ($R_\phi=R_{\mathrm{fine}}-M_\phi$, Theorem~\ref{thm:coarsening_transfer}) exactly under a genuine many-to-one map, with small but non-zero masking mass, showing the identity's practical bite is taxonomy- and model-dependent rather than universal.
Calibration efficiency analysis shows 200 labeled attack samples suffice for near-optimal performance.

The practical implication is that a security automation pipeline should certify capacity, risk, and action rate jointly, not risk alone.
Future work includes searching for a confirmed structural-incapacity instance, extending the coarsening analysis to the full LLM matrix, handling distribution shift via adaptive or non-exchangeable conformal methods, and integrating actionability certification into multi-step incident response pipelines with trajectory-level risk guarantees.

\bibliographystyle{IEEEtran}
\bibliography{refs}

\end{document}